\documentclass[a4paper,11pt,dvipsnames]{article}

\usepackage{amsmath,amsthm,amsfonts,amssymb,bbm,xcolor,mathtools}
\usepackage[margin=2cm]{geometry}
\usepackage{etoolbox}
\usepackage{multicol}
\usepackage{centernot}

\usepackage{svg}
\usepackage{tikz}
\usepackage[only,bigsqcap]{stmaryrd}
\usepackage{subfig}

\usepackage{hyperref}
\usepackage{cleveref}

\usepackage[shortlabels]{enumitem}

\theoremstyle{plain}
\newtheorem{theorem}{Theorem}[section]

\newtheorem{corollary}  [theorem]{Corollary}

\newtheorem{example}    [theorem]{Example}
\newtheorem{lemma}      [theorem]{Lemma}

\newtheorem{proposition}[theorem]{Proposition}

\theoremstyle{definition}
\newtheorem{claim}  {Claim}
\newtheorem*{claim*}{Claim}

\newcommand{\B} {\mathbb{B}}

\newcommand{\Define}    [1] {\textbf{#1}}

\newcommand{\Lattice}   [1] { \mathcal{#1} }

\newcommand{\zero}      { \mathbf{0} }
\newcommand{\one}       { \mathbf{1} }

\newcommand{\id}                    { \mathrm{id} }
\newcommand{\Functions}             { \mathrm{F} }

\newcommand{\Updates}               [1] { \mathtt{#1} }

\newcommand{\Asynchronous}              { \Updates{A} }
\newcommand{\GeneralAsynchronous}       { \Updates{GA} }

\newcommand{\TrappingGraph}             { \Updates{T} }

\newcommand{\Trapping}              [1] { {#1}^\mathrm{T} }
\newcommand{\MinimalTrapping}       [1] { {#1}^\mathrm{M} }

\newcommand{\Trapspaces}            { \mathcal{T} }
\newcommand{\PrincipalTrapspaces}   { \mathcal{P} }
\newcommand{\MinimalTrapspaces}     { \mathcal{M} }
\newcommand{\Subcubes}              { \mathrm{S} }
\newcommand{\Collections}           { \mathrm{A} }

\newcommand{\HammingDistance}{d_\mathrm{H}}

\newcommand{\Fix}{\mathrm{Fix}}

\title{Trapping and commutative Boolean networks}
\author{Maximilien Gadouleau\footnote{Department of Computer Science, Durham University, Durham, UK. \texttt{m.r.gadouleau@durham.ac.uk}} }
\date{}

\begin{document}

\maketitle

\begin{abstract}
A Boolean network is a transformation of the set of Boolean configurations of a given length.
Trapspaces of Boolean networks have been garnering attention due to their theoretical and applicative significance.
A trapspace is a subcube (i.e. a set of configurations defined by fixing certain components) invariant by the Boolean network; a principal trapspace is the smallest trapspace containing a given configuration; a minimal trapspace is one that does not contain any smaller trapspace.
In an unrelated development, commutative Boolean networks have been introduced as those networks where all local updates commute.
In this paper, we relate those two aspects of Boolean network theory. We make five main contributions.
First, we introduce the trapping graph and the trapping closure of a Boolean network. We also define trapping networks as the networks with transitive general asynchronous graphs and we prove that those are exactly the trapping closures.
Second, we show that two Boolean networks have the same collection of (principal) trapspaces if and only if they have the same trapping closure. 
As such, trapping networks are a normal form for the study of trapspaces, for the trapping closure is the network with the most asynchronous transitions for a given collection of trapspaces. 
We then characterise the collections of (principal) trapspaces of Boolean networks. We finally give analogous results for the collections of minimal trapspaces.
Third, we prove that commutative networks are a particular class of trapping networks, and we classify the collections of principal trapspaces of commutative networks.
Fourth, we focus on bijective commutative networks, which we refer to as Marseille networks. We provide several alternative definitions for Marseille networks, and we classify them as special commutative or trapping networks.
Fifth, we focus on idempotent commutative networks, which we refer to as Lille networks. We provide several alternative definitions for Lille networks, we classify them as special commutative or trapping networks, and we relate them to globally idempotent networks. Lille networks are the globally idempotent commutative networks; we then prove that globally idempotent networks are trapping, and we provide alternative definitions for those networks as well.
Our investigations of Marseille and Lille networks also highlight relations amongst the asynchronous, general asynchronous, and trapping graphs of Boolean networks, as well as the structure of trapping networks in general.
\end{abstract}

\section{Introduction} \label{section:introduction}

\subsection{Background} \label{subsection:background}

\paragraph{Boolean networks}

Boolean networks are a fundamental framework for addressing complex systems, with prominent applications in biology, ecology, and social sciences~\cite{Montagud22,Rozum_2024,Gaucherel_2017,Poindron_2021,Grabisch_2013}.
A Boolean network represents a network of $n$ interacting entities, where each entity $i \in [n]$ has a Boolean state $x_i \in \B$, which evolves over time according to a deterministic function $f_i(x_1, \dots, x_n)$ of the current states of the entities.
Mathematically, a Boolean network is simply a mapping $f : \B^n \to \B^n$, which takes an overall configuration of states $x = (x_1, \dots, x_n)$ as input and returns $f(x) = (f_1(x), \dots, f_n(x))$.
Applying $f$ to $x$ corresponds to all entities updating their state at the same time, which is referred to as the parallel schedule.

Of course, different entities may update their state according to different schedules, yielding (general) asynchronous updates. Since the original works by Kauffman \cite{Kau69} and Thomas \cite{Tho73}, asynchronous updates have been widely studied, both in terms of modelling purposes and of dynamical analysis (see \cite{Bor08, ARS23} and references therein). The update of a subset $S \subseteq [n]$ of entities can be represented by $f^{(S)}$, where $f^{(S)}( x ) = (f_S(x), x_{-S})$, and updating $S$ and $T$ successively can be represented by $f^{(S, T)} = f^{(T)} \circ f^{(S)}$. In the fully asynchronous case, only updates of the form $f^{(i)}$ for some $i \in [n]$ occur. All the general asynchronous transitions of the form $x \to f^{(S)}(x)$ are collected in the general asynchronous graph of $f$, while the asynchronous graph of $f$ only considers asynchronous transitions of the form $x \to f^{(i)}(x)$.

\paragraph{Trapspaces}

Arguably the most well studied property of Boolean networks is their fixed points, i.e. $x$ such that $f(x) = x$ (see \cite{Gad20, ARS14, ARS17, Ric15} and references therein).
Beyond their intrinsic significance as ``stable states'' of the modelled network, fixed points have an important theoretical property: they are immune to changes in the update schedule.
Indeed, if $x$ is a fixed point, then $x$ is also a fixed point of $f^{(S)}$ for any $S \subseteq [n]$.

A trapspace of a Boolean network can be viewed as a ``localised'' fixed point, where only part of the network is fixed.
More formally, a trapspace is an invariant subcube: it is a set of the form $X = \{ x \in \B^n : x_S = 0, x_T = 1 \}$ obtained by fixing some states which satisfies $f(X) \subseteq X$, so that those states remain fixed \cite{Klarner15-TrapSpaces}.
Trapspaces have garnered a lot of interest due to their significance in biological applications. Moreover, like fixed points, they are also immune to changes in the update schedule, which makes them theoretically important.

A trapspace is minimal if it does not contain any smaller trapspace. Minimal trapspaces of networks have been notably studied for their relation with limit (ultimately
periodic) configurations: each minimal trapspace necessarily contains at least one limit configuration
\cite{Klarner15-TrapSpaces,pauleve2020reconciling}.

\paragraph{Commutative networks}

Many different classes of Boolean networks have been proposed, based on their interaction graphs (e.g. \cite{Rob80, Aracena2004}), the nature of their local functions (e.g. \cite{Gol85, ARS17}), their metric properties (e.g. \cite{Ric11,Gad19}) etc. Recently, Bridoux et al. \cite{BGT20} introduced the class of commutative networks, where the updates $f^{(S)}$ and $f^{(T)}$ commute for all $S, T \subseteq [n]$. Remarkably, commutative networks are extremely well structured, and their asynchronous graphs have been fully classified. A review of some of the main results in \cite{BGT20} will be provided in Section \ref{subsection:review_commutative}.

\subsection{Contributions} \label{subsection:contributions}

The main scope of this paper brings together the study of trapspaces of networks and that of commutative networks.
The contributions of this paper are broad, as they cover trapping networks, collections of trapspaces, commutative networks, bijective commutative networks, and idempotent commutative networks.
Let us give an overview of each contribution.

Section \ref{section:trapping_networks} is devoted to so-called \textbf{trapping networks}. 
We first introduce the \textit{trapping graph}, which only depends on the principal trapspaces of the network. 
The trapping graph has a handy representation as the general asynchronous graph of another network, which we refer to as the trapping closure. 
We then introduce \textit{trapping networks} as the networks with transitive general asynchronous graphs; a network is trapping if it is the trapping closure of some other network.
Our first main result, Theorem \ref{theorem:trapping_networks}, gives seven different definitions of trapping networks, including the two provided in Table \ref{table:characterisation}.

Section \ref{section:classification_trapspaces} is devoted to \textbf{collections of (principal) trapspaces} of Boolean networks.
We first prove in Theorem \ref{theorem:T(f)=T(g)} that the following are equivalent for two networks: they have the same collection of trapspaces, they have the same collection of principal trapspaces, they have the same trapping graph / closure.
As such, the trapping closure is a ``normal form'' when studying trapspaces: it is the network with the most transitions amongst those with the same collection of trapspaces.
In turn, this shows that when studying trapspaces, one can restrict oneself to trapping networks.
We then give in Theorem \ref{theorem:three-way_equivalence} a \textit{full classification} of the collections of (principal) trapspaces of Boolean networks, and how they relate to one another and to trapping networks.
We finally give a full characterisation of the collections of minimal trapspaces of Boolean networks, and show that two networks with the same collection of minimal trapspaces can have different trapspaces.

Section \ref{section:commutative} is denoted to \textbf{commutative networks}.
Firstly, in Theorem \ref{theorem:commutative_trapping}, we prove that commutative networks are trapping and we provide the alternate definitions of commutative networks in Table \ref{table:characterisation}.
Secondly, in Theorem \ref{theorem:convex_principal_trapspaces}, we classify the collections of principal trapspaces of commutative networks.

Section \ref{section:Marseille_networks} is devoted to bijective commutative networks, which we refer to as \textbf{Marseille networks}\footnote{The terms Marseille and Lille networks come from two facts: (1) Marseille and Lille were the two first locations of the WAN series of workshops; (2) bijective commutative networks were first discussed in Marseille by the authors in \cite{BGT20}, while idempotent commutative networks were first exposed at the WAN workshop in Lille.}.
We make three main contributions with regards to Marseille networks.
\begin{itemize}
    \item Firstly, we give four alternate definitions of Marseille networks, including those in Table \ref{table:characterisation}, in Theorem \ref{theorem:Marseille_definitions}. 

    \item Secondly, we study the properties of networks with symmetric asynchronous / general asynchronous / trapping graphs and relate them to involutive networks in Theorem \ref{theorem:graphs_symmetric}.
    In particular, the following are equivalent for a network: it is Marseille; it is globally involutive (i.e. $f^{(S,S)} = \id$ for all $S \subseteq [n]$); its general asynchronous graph is symmetric.

    \item Thirdly, we characterise Marseille networks as particular trapping or commutative networks in Theorem \ref{theorem:Marseille_classification}.
    In particular, all locally bijective (i.e. $f^{(i)}$ is a bijection for all $i \in [n]$) trapping networks are Marseille.
\end{itemize}

Section \ref{section:Lille} is devoted to idempotent commutative networks, which we refer to as \textbf{Lille networks}.
We make four main contributions with regards to Lille networks.
\begin{itemize}
    \item Firstly, we give four alternate definitions of Lille networks, including those in Table \ref{table:characterisation}, in Theorem \ref{theorem:Lille_definitions}.

    \item Secondly, we study the properties of networks with triangular (i.e. acyclic with loops) or oriented asynchronous / general asynchronous / trapping graphs and relate them to idempotent networks in Theorem \ref{theorem:graphs_triangular}.
    In particular, a trapping network is fixable (i.e. its asynchronous attractors are fixed points) if and only if every trapspace contains a fixed point.

    \item Thirdly, we consider the fifth main class of networks in this paper, namely \textit{globally idempotent networks}.
    Lille networks are exactly the commutative globally idempotent networks.
    In Theorem \ref{theorem:globally_idempotent_definitions}, we prove that globally idempotent networks are trapping and we provide the alternate definitions of globally idempotent networks in Table \ref{table:characterisation}.

    \item Fourthly, we characterise Lille networks as particular trapping or commutative networks in Theorem \ref{theorem:Lille_classification}.
    In particular, a network is fixable if for any configuration $x$, there is a path in the asynchronous graph from $x$ to a fixed point.
    We show that Lille networks are exactly the fixable commutative networks.

\end{itemize}

\begin{table}
\centering
\begin{tabular}{c|c|c}
    & $\forall x \in \B^n, y \in [x, f(x) ]$ & $\forall S, T \subseteq [n]$ \\
    \hline
    \textcolor{magenta}{Trapping}          &  $\textcolor{magenta}{ [y, f(y)] } \subseteq [ x, f(x) ]$ & $\textcolor{magenta}{ f^{(S,T)} } \sqsubseteq f^{( S \cup T )}$ \\
    \textcolor{blue}{Commutative}   & \;\; $[y, f(x)] \subseteq \textcolor{blue}{ [y, f(y)] } \subseteq [ x, f(x) ]$ \;\; & \;\; $f^{( S \Delta T )} \sqsubseteq \textcolor{blue}{ f^{(S,T)} } \sqsubseteq f^{( S \cup T )}$ \;\;\\
    \textcolor{cyan}{Marseille}     &  $\textcolor{cyan}{ [y, f(y)] } = [ x, f(x) ]$ & $\textcolor{cyan}{ f^{(S,T)} } = f^{( S \Delta T )}$ \\
    \textcolor{red}{Lille}           &  $\textcolor{red}{ [ y, f(y) ] } = [ y, f(x) ]$ & $\textcolor{red}{ f^{(S,T)} } = f^{( S \cup T )}$\\
    \;\; \textcolor{orange}{Globally Idempotent} \;\; & $\textcolor{orange}{ [y, f(y)] } \subseteq [ y, f(x) ]$ & $f^{( S \cap T )} \sqsubseteq \textcolor{orange}{  f^{(S,T)} }$   
\end{tabular}
\caption{Characterisation of the five main classes of Boolean networks in this paper.} \label{table:characterisation}
\end{table}

\section{Preliminaries} \label{section:preliminaries}

\paragraph{Boolean configurations and subcubes}

We denote the Boolean set by $\B = \{0,1\}$ and for any positive integer $n$, we denote $[n] = \{1, \dots, n \}$. A \Define{configuration} is $x = (x_1, \dots, x_n) \in \B^n$. For any $S \subseteq [n]$, we denote $x_S = (x_s : s \in S)$ and $x_{-S} = ( x_t : t \notin S )$, and we use the notation $x = (x_S, x_{-S})$. We shall identify an element $i \in [n]$ with the corresponding singleton $\{ i \}$, so that $x = (x_i, x_{-i})$ for instance. For any two configurations $x, y \in \B^n$, we denote the set of positions where they differ by $\Delta( x, y ) = \{ i \in [n] : x_i \ne y_i \}$ and their Hamming distance by $\HammingDistance( x, y ) = | \Delta(x, y) |$. For any Boolean variable $a \in \B$, we denote its negation by $\neg a = 1 - a$; we extend this notation to configurations of any length by componentwise negation: $\neg x = ( \neg x_1, \dots, \neg x_n )$.

A \Define{subcube} of $\B^n$ is any $X \subseteq \B^n$ such that there exist two disjoint sets of positions $S, T \subseteq [n]$ with $X = \{ x \in \B^n : x_S = 0, x_T = 1 \}$. For any set $A \subseteq \B^n$, the principal subcube of $A$, denoted by $[A]$, is the smallest subcube containing $A$. If $A = \{a_1, \dots, a_k\}$, we also denote $[A] = [a_1, \dots, a_k]$. If $X$ is a subcube and $x \in X$, then there is a unique $y \in X$ such that $X = [x, y]$; we refer to $y$ as the \Define{opposite} of $x$ in $X$, and we denote it by $y = X - x$. We denote the set of subcubes of $\B^n$ by $\Subcubes(n)$ and the set of all collections of subcubes of $\B^n$ by $\Collections(n) = 2^{ \Subcubes(n) }$.

\paragraph{Boolean networks}

A \Define{Boolean network} (or simply, network) of dimension $n$ is a mapping $f : \B^n \to \B^n$. We denote the set
of networks of dimension $n$ as $\Functions(n)$.
For any $x \in \B^n$, we refer to the subcube $[x, f(x)]$ as the \Define{interval} of $x$ with respect to $f$.
Any network $f \in \Functions(n)$ can be viewed as $f = (f_1, \dots, f_n)$ where $f_i : \B^n \to \B$ is given by $f_i(x) = f(x)_i$ for
all $i \in [n]$. For any $S \subseteq [n]$ and any $f \in \Functions(n)$, the update of $S$
according to $f$ is represented by the network $f^{(S)} \in \Functions(n)$ where
\[
    f^{ (S) }( x ) = ( f_S( x ), x_{-S} ).
\]
In particular, $f^{ ([n]) } = f$ and $f^{ (\emptyset) } = \id$. We note the distinction between the update $f^{ (i) }$ (given by $f^{ (i) }(x) = ( f_i( x ), x_{-i} )$) and the power $f^i = f \circ f \circ \dots \circ f$ ($i$ terms). We can then compose those updates, so that if $S_1, \dots, S_k \subseteq [n]$, we obtain
\[
    f^{( S_1, \dots, S_k )} = f^{(S_k)} \circ f^{(S_{k-1})} \circ \dots \circ f^{(S_1)}.
\]

\paragraph{Asynchronous graph and general asynchronous graph}

A (directed) \Define{graph} is $\Gamma  = (V, E)$, where $V$ is the set of vertices and $E \subseteq V^2$ is the set of edges. A graph $\Gamma$ is \Define{reflexive} if for all $v \in V$, $(v,v) \in E$; $\Gamma$ is \Define{symmetric} if $(u,v) \in E$ implies $(v,u) \in E$ for all $u,v \in V$; and $\Gamma$ is \Define{transitive} if $(u,v), (v,w) \in E$ implies $(u,w) \in E$ for all $u,v,w \in V$. The \Define{out-neighbourhood} of a vertex $v$ is $N^{out}( \Gamma ; v) = \{ u \in V : (v,u) \in E \}$.

The \Define{asynchronous graph} of a network $f \in \Functions(n)$ is the graph $\Asynchronous(f) = (V, E)$ where $V = \B^n$ and 
\[
    E = \{ (x, f^{ (i) }(x) ) : x \in \B^n, i \in [n] \}. 
\]
We remark that a Boolean network is fully characterised by its asynchronous graph.
Note that in most literature, one removes the loops $(x, x)$ from the asynchronous graph, that occur every time $f_i(x) = x_i$, yet we shall keep those loops instead in our definition. However, when drawing the asynchronous graph, we shall not display the loops and instead draw the underlying hypercube with thin black lines and the arcs of the graph with thick \textcolor{blue}{blue} arrows. See below for an example of a network, for which we give the asynchronous graph.

\begin{example} \label{example:network}
Consider the Boolean network $f \in \Functions( 3 )$.\\
~\\
\begin{tabular}{c|c}
     $x$ & $f(x)$ \\
     \hline
     $000$ & $110$ \\ 
     $001$ & $100$ \\
     $010$ & $000$ \\
     $011$ & $110$ \\
     $100$ & $100$ \\
     $101$ & $101$ \\
     $110$ & $110$ \\
     $111$ & $110$ 
\end{tabular}
~\\
The asynchronous graph of $f$ is given as follows.

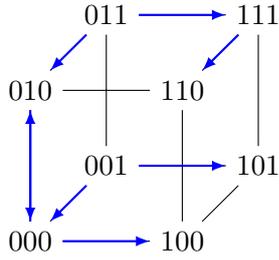
\begin{center}
\begin{tikzpicture}[scale=1.5]

    \node (000) at (0,0) {$000$};
    \node (001) at (1,1) {$001$};
    \node (010) at (0,2) {$010$};
    \node (011) at (1,3) {$011$};
    \node (100) at (2,0) {$100$};
    \node (101) at (3,1) {$101$};
    \node (110) at (2,2) {$110$};
    \node (111) at (3,3) {$111$};

    \path[draw] (000) -- (001) -- (011) -- (111)
    (000) -- (010) -- (110) -- (111)
    (000) -- (100) -- (101) -- (111)
    (001) -- (101)
    (010) -- (011)
    (100) -- (110);
    
    \draw[very thick,-latex, blue] (000) -- (100);
    \draw[very thick,-latex, blue] (000) -- (010);

    \draw[very thick,-latex, blue] (001) -- (101);
    \draw[very thick,-latex, blue] (001) -- (000);

    \draw[very thick,-latex, blue] (010) -- (000);

    \draw[very thick,-latex, blue] (011) -- (111);
    \draw[very thick,-latex, blue] (011) -- (010);

    \draw[very thick,-latex, blue] (111) -- (110);
\end{tikzpicture}
\end{center}

\end{example}

The \Define{general asynchronous graph} of a network $f \in \Functions(n)$ is the graph $\GeneralAsynchronous(f) = (V, E)$ where $V = \B^n$ and 
\[
    E = \{ (x, f^{ (S) }(x) ) : x \in \B^n, S \subseteq [n] \}.
\]
Equivalently, the out-neighbourhood of a configuration in the general asynchronous graph is given by its interval: $N^{out}(\GeneralAsynchronous(f); x) = [x, f(x)]$ for all $x \in \B^n$. It is clear that the mapping $\GeneralAsynchronous : f \mapsto \GeneralAsynchronous(f)$ is injective. We first characterise the general asynchronous graphs.

\begin{proposition} \label{proposition:general_asynchronous_graph}
Let $\Gamma$ be a graph on $\B^n$. Then $\Gamma = \GeneralAsynchronous(f)$ for some $f \in \Functions(n)$ if and only if $\Gamma$ is reflexive and all out-neighbourhoods are subcubes.
\end{proposition}

\begin{proof}
We have $N^{out}(\GeneralAsynchronous(f); x) = [x, f(x)]$ for all $x$, i.e. $\GeneralAsynchronous(f)$ is reflexive and all out-neighbourhoods are subcubes. Conversely, if $\Gamma$ is reflexive and all out-neighbourhoods are subcubes, then $\Gamma = \GeneralAsynchronous(f)$, where $f(x) = N^{out}( \Gamma; x ) - x$ for all $x$.
\end{proof}

\begin{example} \label{example:GA}
Let $f \in \Functions(3)$ be the network in Example \ref{example:network}. The general asynchronous graph of $f$ is given as follows. The \textcolor{blue}{blue} arrows come from $\Asynchronous(f)$ while the \textcolor{magenta}{magenta} arrows are additional transitions in $\GeneralAsynchronous(f)$; once again we omit the loops on all the vertices.

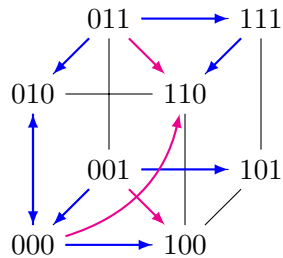
\begin{center}
\begin{tikzpicture}[scale=1.5]
    \node (000) at (0,0) {$000$};
    \node (001) at (1,1) {$001$};
    \node (010) at (0,2) {$010$};
    \node (011) at (1,3) {$011$};
    \node (100) at (2,0) {$100$};
    \node (101) at (3,1) {$101$};
    \node (110) at (2,2) {$110$};
    \node (111) at (3,3) {$111$};

    \path[draw] (000) -- (001) -- (011) -- (111)
    (000) -- (010) -- (110) -- (111)
    (000) -- (100) -- (101) -- (111)
    (001) -- (101)
    (010) -- (011)
    (100) -- (110);
    
    \draw[very thick,-latex, blue] (000) -- (100);
    \draw[very thick,-latex, blue] (000) -- (010);

    \draw[very thick,-latex, blue] (001) -- (101);
    \draw[very thick,-latex, blue] (001) -- (000);

    \draw[very thick,-latex, blue] (010) -- (000);

    \draw[very thick,-latex, blue] (011) -- (111);
    \draw[very thick,-latex, blue] (011) -- (010);

    \draw[very thick,-latex, blue] (111) -- (110);

    \draw[very thick,-latex, magenta] (000) to [bend right]  (110);
    \draw[very thick,-latex, magenta] (001) -- (100);
    \draw[very thick,-latex, magenta] (011) -- (110);
    
\end{tikzpicture}
\end{center}
\end{example}

\paragraph{The lattice of Boolean networks}

The networks in $\Functions( n )$ have a natural partial order in terms of their transitions. For any two graphs $\Gamma_1 = (V, E_1)$, $\Gamma_2 = (V, E_2)$ on the same vertex $V$, we write $\Gamma_1 \subseteq \Gamma_2$ as a shorthand for $E_1 \subseteq E_2$.

Consider the relation $f \sqsubseteq g$ on $\Functions(n)$ given by the four equivalent conditions:
\begin{itemize}
    \item 
    $\Delta(x, f(x)) \subseteq \Delta(x, g(x))$ for all $x \in \B^n$;

    \item 
    $[ x, f(x) ] \subseteq [ x, g(x) ]$ for all $x \in \B^n$;

    \item 
    $\GeneralAsynchronous(f) \subseteq \GeneralAsynchronous(g)$;

    \item 
    $\Asynchronous( f ) \subseteq \Asynchronous( g )$.
\end{itemize}
The $\sqsubseteq$ partial order induces a lattice $\Lattice{F}_n = ( \Functions(n), \sqcup, \sqcap, \zero, \one)$ (a Boolean algebra isomorphic to $2^E$, where $E = \{ (u,v) : u,v \in \B^n, \HammingDistance(u,v) = 1 \}$ is the set of arcs of the hypercube) with 
\begin{align*}
    \GeneralAsynchronous( f \sqcup g ) &= \GeneralAsynchronous(f) \cup \GeneralAsynchronous(g)\\
    \GeneralAsynchronous( f \sqcap g ) &= \GeneralAsynchronous(f) \cap \GeneralAsynchronous(g)\\
    \zero &= \id \quad (\text{the identity network: } f(x) = x)\\
    \one &= \neg \quad (\text{the negation network: } f(x) = \neg x).
\end{align*}

We note that the ordering behaves well when considering updates of subsets: $f^{(S)} \sqsubseteq f^{(S \cup T)}$ for all $S, T \subseteq [n]$ and $f \sqsubseteq g \implies f^{(S)} \sqsubseteq g^{(S)}$ for all $S \subseteq [n]$.

\paragraph{Trapspaces of Boolean networks}

A \Define{trapspace} of $f \in \Functions(n)$ is a subcube $X \subseteq \B^n$ that satisfies the following three equivalent conditions \cite{Klarner15-TrapSpaces}:
\begin{itemize}
    \item 
    $f(X) \subseteq X$;

    \item 
    $f^{(i)}( X ) \subseteq X$ for all $i \in [n]$;

    \item 
    $f^{(S)}( X ) \subseteq X$ for all $S \subseteq [n]$.
\end{itemize}

The collection of all trapspaces of $f$, denoted by $\Trapspaces(f)$, is closed under intersection. Then for any $x \in \B^n$, there is a smallest trapspace of $f$ that contains $x$, which we shall refer to as the \Define{principal trapspace} of $x$ (with respect to $f$). For the sake of simplicity, we denote it by $T_f(x)$. The principal trapspace $T_f(x)$ can be recursively computed as follows: let $T_0 = \{ x \}$ and $T_i = [T_{i-1} \cup f( T_{i-1} )]$, then $T_n = T_f(x)$. In particular, the principal trapspace of $x$ contains the interval of $x$: $[x, f(x)] \subseteq T_f(x)$. The collection of all principal trapspaces of $f$ is denoted by $\PrincipalTrapspaces(f)$. A trapspace $T$ is \Define{minimal} if there is no trapspace $T'$ with $T' \subset T$. Clearly, any minimal trapspace is principal, but the converse does not necessarily hold. The collection of minimal trapspaces of $f$ is denoted by $\MinimalTrapspaces(f)$. 

\begin{example} \label{example:trapspaces}
Let $f \in \Functions(3)$ be the network from Examples \ref{example:network} and \ref{example:GA}. The principal trapspaces of $f$ are given as follows:
\begin{align*}  
    T_f(000) &= [000,110] = \{ x : x_3 = 0 \} \\
    T_f(001) &= [001,110] = \{ x \} \\
    T_f(010) &= [010,100] = \{ x : x_3 = 0 \} \\
    T_f(011) &= [011,100] = \{ x \} \\
    T_f(100) &= [100,100] = \{ x : x_1 = 1, x_{23} = 00 \} \\
    T_f(101) &= [101,101] = \{ x : x_{13} = 11, x_2 = 0 \} \\
    T_f(110) &= [110,110] = \{ x : x_{12} = 11, x_3 = 0 \} \\
    T_f(111) &= [111,110] = \{ x : x_{12} = 11 \}.
\end{align*}
It has three other trapspaces, namely $\{ x : x_1 = 1, x_2 = 0 \}$, $\{ x : x_1 = 1, x_3 = 0 \}$, and $\{ x : x_1 = 1 \}$.
\end{example}

\section{Trapping networks} \label{section:trapping_networks}

\subsection{Trapping graph and trapping closure} \label{subsection:trapping_graph}

The \Define{trapping graph} of a network $f \in \Functions(n)$ is the graph $\TrappingGraph(f) = (V, E)$ where $V = \B^n$ and
\[
    E = \{ (x,y) : x \in \B^n, y \in T_f( x ) \}.
\]
If $y \in T_f(x)$, then $T_f(x)$ is a trapspace of $f$ containing $y$, thus $T_f(y) \subseteq T_f(x)$; in other words, the trapping graph is transitive. In $\TrappingGraph(f)$, the out-neighbourhood of $x$ is a subcube containing $x$; therefore, the trapping graph of $f$ is the general asynchronous graph of another network. More concretely, let $\Trapping{f} \in \Functions(n)$ be the network that maps $x$ to its opposite in that subcube, i.e. 
\[
    \Trapping{f}(x) = T_f(x) - x,
\]
so that $T_f(x) = [x, \Trapping{f}(x)]$ and $\GeneralAsynchronous( \Trapping{f} ) = \TrappingGraph( f )$. We refer to $\Trapping{f}$ as the \Define{trapping closure} of $f$.

\begin{example} \label{example:trapping_graph}
Let $f \in \Functions(3)$ be the network from Examples \ref{example:network}, \ref{example:GA} and \ref{example:trapspaces}. The trapping graph of $f$ is given as follows.  The \textcolor{blue}{blue} arrows come from $\Asynchronous(f)$, the \textcolor{magenta}{magenta} arrows are additional transitions in $\GeneralAsynchronous(f)$, while the \textcolor{orange}{orange} arrows are additional transitions in $\TrappingGraph(f)$; once again we omit the loops on all the vertices.

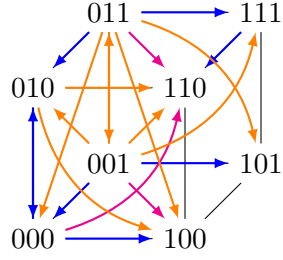
\begin{center}
\begin{tikzpicture}[scale=1.5]
    \node (000) at (0,0) {$000$};
    \node (001) at (1,1) {$001$};
    \node (010) at (0,2) {$010$};
    \node (011) at (1,3) {$011$};
    \node (100) at (2,0) {$100$};
    \node (101) at (3,1) {$101$};
    \node (110) at (2,2) {$110$};
    \node (111) at (3,3) {$111$};

    \path[draw] (000) -- (001) -- (011) -- (111)
    (000) -- (010) -- (110) -- (111)
    (000) -- (100) -- (101) -- (111)
    (001) -- (101)
    (010) -- (011)
    (100) -- (110);
    
    \draw[very thick,-latex, blue] (000) -- (100);
    \draw[very thick,-latex, blue] (000) -- (010);

    \draw[very thick,-latex, blue] (001) -- (101);
    \draw[very thick,-latex, blue] (001) -- (000);

    \draw[very thick,-latex, blue] (010) -- (000);

    \draw[very thick,-latex, blue] (011) -- (111);
    \draw[very thick,-latex, blue] (011) -- (010);

    \draw[very thick,-latex, blue] (111) -- (110);

    \draw[very thick,-latex, magenta] (000) to [bend right]  (110);
    \draw[very thick,-latex, magenta] (001) --  (100);
    \draw[very thick,-latex, magenta] (011) --  (110);

    \draw[very thick,-latex, orange] (010) -- (110);
    \draw[very thick,-latex, orange] (010) to [bend right]  (100);



    \draw[very thick,-latex, orange] (001) -- (010);
    \draw[very thick,-latex, orange] (001) --  (011);
    \draw[very thick,-latex, orange] (001) -- (110);
    \draw[very thick,-latex, orange] (001) to [bend right]  (111);

    \draw[very thick,-latex, orange] (011) --  (000);
    \draw[very thick,-latex, orange] (011) --  (001);
    \draw[very thick,-latex, orange] (011) --  (100);
    \draw[very thick,-latex, orange] (011) to [bend left]  (101);

\end{tikzpicture}
\end{center}

The asynchronous graph of $\Trapping{f}$ is given as follows, where the \textcolor{blue}{blue} arrows come from $\Asynchronous(f)$, while the additional transitions in $\Asynchronous( \Trapping{f} )$ are highlighted in \textcolor{violet}{violet}.

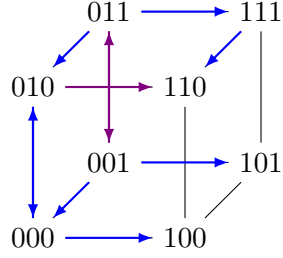
\begin{center}
\begin{tikzpicture}[scale=1.5]
    \node (000) at (0,0) {$000$};
    \node (001) at (1,1) {$001$};
    \node (010) at (0,2) {$010$};
    \node (011) at (1,3) {$011$};
    \node (100) at (2,0) {$100$};
    \node (101) at (3,1) {$101$};
    \node (110) at (2,2) {$110$};
    \node (111) at (3,3) {$111$};

    \path[draw] (000) -- (001) -- (011) -- (111)
    (000) -- (010) -- (110) -- (111)
    (000) -- (100) -- (101) -- (111)
    (001) -- (101)
    (010) -- (011)
    (100) -- (110);
    
    \draw[very thick,-latex, blue] (000) -- (100);
    \draw[very thick,-latex, blue] (000) -- (010);

    \draw[very thick,-latex, blue] (001) -- (101);
    \draw[very thick,-latex, violet] (001) -- (011);
    \draw[very thick,-latex, blue] (001) -- (000);

    \draw[very thick,-latex, violet] (010) -- (110);
    \draw[very thick,-latex, blue] (010) -- (000);

    \draw[very thick,-latex, blue] (011) -- (111);
    \draw[very thick,-latex, violet] (011) -- (001);
    \draw[very thick,-latex, blue] (011) -- (010);

    \draw[very thick,-latex, blue] (111) -- (110);
\end{tikzpicture} 
\end{center}
\end{example}

\begin{proposition} \label{proposition:trapping_closure}
The operator $f \mapsto \Trapping{f}$ is a closure operator on the lattice $\Lattice{F}_n$: for all $f, g \in \Functions(n)$ we have
\begin{align}
    \label{equation:trapping1}
    f &\sqsubseteq \Trapping{f}\\
    \label{equation:trapping2}
    f \sqsubseteq g &\implies \Trapping{f} \sqsubseteq \Trapping{g}\\
    \label{equation:trapping3}
    \Trapping{ (\Trapping{f}) } &= \Trapping{f}.
\end{align}
\end{proposition}

\begin{proof}
\eqref{equation:trapping1}. Since $f(x) \in T_f(x)$ for all $x$, we obtain $f \sqsubseteq \Trapping{f}$.

\eqref{equation:trapping2}. Suppose $f \sqsubseteq g$, then we need to prove that $T_f(x) \subseteq T_g(x)$ for all $x \in \B^n$. For all $y \in T_g(x)$, we have $[y, g(y)] \subseteq T_g(x)$ hence $[y, f(y)] \subseteq T_g(x)$. Thus $T_g(x)$ is a trapspace of $f$ containing $x$ whence $T_f(x) \subseteq T_g(x)$. 

\eqref{equation:trapping3}. Let $g = \Trapping{f}$. We prove that $T_g(x) = T_f(x)$ for all $x \in \B^n$. Firstly, $T_f(x) \subseteq T_g(x)$ because $f \sqsubseteq g$. Conversely, for any $y \in T_f(x)$, we have $T_f(y) = [y, g(y)] \subseteq T_f(x)$, hence $g(y) \in T_f(x)$. Therefore $T_f(x)$ is a trapspace of $g$ containing $x$, thus $T_g(x) \subseteq T_f(x)$.
\end{proof}

We can now give an algebraic characterisation of $\Trapping{f}$.

\begin{corollary} \label{corollary:trapping_closure}
For all $f \in \Functions(n)$ we have
\[
    \Trapping{f} = \bigsqcap \{ \Trapping{g} : f \sqsubseteq g\}.
\]
\end{corollary}

\begin{proof}
Immediately follows from \cite[Theorem 1.1]{Coh65} applied to $\Lattice{F}_n$.
\end{proof}

\subsection{Trapping networks} \label{subsection:trapping_networks}

We say a network is \Define{trapping} if its general asynchronous graph is transitive. Denote the set of all trapping networks in $\Functions(n)$ by $\Trapping{\Functions}(n)$. We now provide a list of equivalent definitions of trapping networks.

Many proofs of our results will be broken down into smaller proofs of the form ``Property $A$ $\implies$ Property $B$''. For all such proofs, we omit the introductory sentence: ``Let $f \in \Functions( n )$ satisfy $A$.''

\begin{theorem}[Alternate definitions of trapping networks] \label{theorem:trapping_networks}
For all $f \in \Functions(n)$, the following are equivalent:
\begin{enumerate}
	\item \label{item:g_trapping}
    $f$ is trapping, i.e. $\GeneralAsynchronous( f )$ is transitive,

    \item \label{item:delta(y,g(y))}
    $[ y, f(y) ] \subseteq [ x, f(x) ]$ for all $x \in \B^n$ and $y \in [x, f(x) ]$,
    
    \item \label{item:trapping_interval=Tf}
    $[x, f(x)] = T_f( x )$ for all $x \in \B^n$,
    
    \item \label{item:g=gt}
    $f = \Trapping{f}$,
	
	\item \label{item:g=ft}
    $f = \Trapping{g}$ for some $g \in \Functions( n )$,

    \item \label{item:TG(g)=GA(g)}
    $\TrappingGraph( f ) = \GeneralAsynchronous( f )$.

    \item \label{item:trapping_ST}
    $f^{(S,T)} \sqsubseteq f^{(S \cup T)}$ for all $S, T \subseteq [n]$.

\end{enumerate}
\end{theorem}

\begin{proof}
$\ref{item:g_trapping} \iff \ref{item:delta(y,g(y))}$. By definition.

$\ref{item:delta(y,g(y))} \implies \ref{item:trapping_interval=Tf}$. If $f(y) \in [x, f(x)]$ for all $y \in [x, f(x)]$, then $[x, f(x)]$ is a trapspace. Since $T_f( x )$ is the smallest trapspace that contains $x$, and $[x, f(x)] \subseteq T_f( x )$, we obtain $[x, f(x)] = T_f( x )$.

$\ref{item:trapping_interval=Tf} \implies \ref{item:delta(y,g(y))}$. If $[x, f(x)]$ is a trapspace, then $f(y) \in [x, f(x)]$ for all $y \in [x, f(x)]$.

$\ref{item:trapping_interval=Tf} \iff \ref{item:g=gt}$. By definition.

$\ref{item:g=gt} \implies \ref{item:g=ft}$. Trivial.

$\ref{item:g=ft} \implies \ref{item:g=gt}$. Follows directly from \eqref{equation:trapping3} in Proposition \ref{proposition:trapping_closure}. 

$\ref{item:g=gt} \implies \ref{item:TG(g)=GA(g)}$. If $f = \Trapping{f}$, then $\GeneralAsynchronous( f ) = \GeneralAsynchronous( \Trapping{f} ) = \TrappingGraph( f )$.

$\ref{item:TG(g)=GA(g)} \implies \ref{item:g_trapping}$. If $\GeneralAsynchronous( f ) = \TrappingGraph( f )$, then $\GeneralAsynchronous( f )$ is indeed transitive.

$\ref{item:trapping_interval=Tf} \implies \ref{item:trapping_ST}$. Suppose $x \to z$ in $\GeneralAsynchronous( f^{(S,T)} )$. Denoting $s = f^{(S)}( x )$ and $t = f^{(T)}(s)$, we have $s \in [x, f(x)] = T_f( x )$, $t \in T_f( s ) \subseteq T_f( x )$, and $z \in [x, t] \subseteq T_f( x ) = N^{out}( \GeneralAsynchronous(f) ; x )$. Since $\Delta( x, z ) \subseteq \Delta(x, t) \subseteq S \cup T$, we obtain $z \in N^{out}( \GeneralAsynchronous( f^{(S \cup T)} ) ; x )$. Thus, $f^{( S,T )} \sqsubseteq f^{( S \cup T )}$.

$\ref{item:trapping_ST} \implies \ref{item:g_trapping}$. Let $x \to y \to z$ in $\GeneralAsynchronous( f )$, so that $y = f^{(S)}( x )$ and $z = f^{(T)} (y) = f^{(S, T)} (x)$ for some $S, T \subseteq [n]$. We obtain $x \to z$ in $\GeneralAsynchronous( f^{(S,T)} ) \subseteq \GeneralAsynchronous( f^{(S \cup T)} ) \subseteq \GeneralAsynchronous( f )$. Therefore, $\GeneralAsynchronous( f )$ is transitive.
\end{proof}

Theorem \ref{theorem:trapping_networks} yields the following corollary.

\begin{corollary} \label{corollary:TG(f)=GA(fP)}
For all $f \in \Functions(n)$,
\[
    \TrappingGraph(f) = \GeneralAsynchronous( \Trapping{f} ) = \TrappingGraph( \Trapping{f} ).
\]
\end{corollary}

We can also classify the trapping graphs as the transitive general asynchronous graphs. 

\begin{corollary} \label{corollary:classification_trapping_graphs}
Let $\Gamma$ be a graph on $\B^n$. Then $\Gamma = \TrappingGraph(f)$ for some $f \in \Functions(n)$ if and only if $\Gamma$ is reflexive transitive and all out-neighbourhoods are subcubes.
\end{corollary}

\begin{proof}
By Proposition \ref{proposition:general_asynchronous_graph}, $\Gamma$ is reflexive transitive and all out-neighbourhoods are subcubes if and only if it is a transitive general asynchronous graph, i.e. $\Gamma = \GeneralAsynchronous(g)$ for some trapping network $g$. By Corollary \ref{corollary:TG(f)=GA(fP)}, this is equivalent to $\Gamma = \TrappingGraph(f)$ for some network $f$.
\end{proof}

Trapping graphs form a rich class of graphs. For instance, any $X \subseteq \B^n$ can appear as an initial strong component of some trapping graph (namely, for $f(x) = \neg x$ if $x \in X$ and $f(x) = x$ otherwise). Note, however, that if $\TrappingGraph(f)$ has two distinct strong components $S, T$ with $S \to T$, then $[T] \subset [S]$.

\section{Collections of (principal) trapspaces} \label{section:classification_trapspaces}

\subsection{Boolean networks with the same collection of (principal) trapspaces}

We begin this section with a characterisation of networks that have the same collection of trapspaces. In particular, Theorem \ref{theorem:T(f)=T(g)} below shows that trapping networks are a canonical form for networks when studying their trapspaces. Recall that the collection of all trapspaces of $f$ is denoted by $\Trapspaces(f)$, while the collection of all principal trapspaces of $f$ is denoted by $\PrincipalTrapspaces(f)$. 


\begin{theorem}[Trapspace equivalent networks] \label{theorem:T(f)=T(g)}
Let $f, g \in \Functions(n)$. The following are equivalent:
\begin{enumerate}
    \item \label{item:T(f)=T(g)}
    $f$ and $g$ have the same collection of principal trapspaces, i.e. 
    $\PrincipalTrapspaces(f) = \PrincipalTrapspaces(g)$;

    \item \label{item:tau(f)=tau(g)}
    $f$ and $g$ have the same collection of trapspaces, i.e. 
    $\Trapspaces(f) = \Trapspaces(g)$;

    \item \label{item:Tf(x)=Tg(x)}
    $f$ and $g$ have the same principal trapspaces pointwise, i.e. 
    $T_f(x) = T_g(x)$ for all $x \in \B^n$;

    \item \label{item:TG(f)=TG(g)}
    $f$ and $g$ have the same trapping graph, i.e. 
    $\TrappingGraph(f) = \TrappingGraph(g)$;

    \item \label{item:fT=gT}
    $f$ and $g$ have the same trapping closure, i.e. 
    $\Trapping{f} = \Trapping{g}$.
\end{enumerate}
\end{theorem}

\begin{proof}
$\ref{item:T(f)=T(g)} \implies \ref{item:Tf(x)=Tg(x)}$. Let $x \in \B^n$. On the one hand, since $T_f(x)$ and $T_g(x)$ are trapspaces of $f$ containing $x$, we have $T_f(x) \subseteq T_g(x)$. We similarly obtain $T_g(x) \subseteq T_f(x)$, and hence $T_f(x) = T_g(x)$.

$\ref{item:Tf(x)=Tg(x)} \implies \ref{item:tau(f)=tau(g)}$. We have
\[
    A \in \Trapspaces(f) \iff A = \bigcup \{ T_f(x) : x \in A \} \iff A = \bigcup \{ T_g(x) : x \in A \} \iff A \in \Trapspaces(g).
\]

$\ref{item:tau(f)=tau(g)} \implies \ref{item:T(f)=T(g)}$. Trivial.

$\ref{item:Tf(x)=Tg(x)} \iff \ref{item:TG(f)=TG(g)} \iff \ref{item:fT=gT}$. Immediate from the definitions of $\TrappingGraph(f)$ and $\Trapping{f}$.
\end{proof}

\begin{corollary} \label{corollary:equal_collections_of_trapspaces}
For any network $f$, $\PrincipalTrapspaces(f) = \PrincipalTrapspaces( \Trapping{f} )$ and $\Trapspaces( f ) = \Trapspaces( \Trapping{f} )$.
\end{corollary}

\subsection{Classification of collections of (principal) trapspaces} \label{subsection:classification_theorem}

\paragraph{Intuition for the classification}

In this section, we shall classify the collections of trapspaces and the collections of principal trapspaces of networks. We shall moreover illustrate a three-way equivalence amongst collections of principal trapspaces, collections of trapspaces, and trapping closures. This equivalence is similar to the situation for pre-orders on sets.

A \Define{pre-order} on a set $\Omega$ is a reflexive transitive binary relation on $\Omega$. If $R$ is a pre-order on $\Omega$, then any set of the form $S^\downarrow = \{ y \in \Omega : \exists s \in S, (s,y) \in R \}$ for some $S \subseteq \Omega$ is an \Define{ideal} of $R$; a \Define{principal ideal} of $R$ is any set of the form $x^\downarrow = \{ y \in \Omega, (x,y) \in R \}$ for some $x \in \Omega$. Since $R = \bigcup_{x \in \Omega, y \in x^\downarrow} (x,y)$, we see that $R$ can be reconstructed from its collection of principal ideals; the same can be said for its collection of ideals. It is well known that a collection of subsets of $\Omega$ is the collection of ideals of a pre-order if and only if it is closed under arbitrary unions and intersections; similarly one can classify the collections of principal ideals of pre-orders. 
Therefore, for the set $\Omega$, there is a three-way equivalence between a pre-order $R$ on $\Omega$, its collection of ideals, and its collection of principal ideals.

In this paper, we are interested in $\Omega = \B^n$, but we do not consider any possible (principal) ideal. Let $f \in \Functions(n)$ be a network. The reachability relation $R$ given by $R = \{ (x,y) :  x \to_{ \GeneralAsynchronous(f) } \dots \to_{ \GeneralAsynchronous(f) } y \}$ is a pre-order on $\B^n$. A subcube is an ideal of $R$ if and only if it is a trapspace of $f$; it is a principal ideal of $R$ if and only if it is a principal trapspace of $f$. The relation $R'$ given by $y \in T_f(x)$ is also a pre-order, described by the trapping graph ($(x,y) \in R'$ if and only if $(x,y)$ is an edge of $\TrappingGraph(f)$). This is the smallest pre-order such that all its principal ideals are principal trapspaces of $f$. 
Therefore, the main result is a three-way equivalence between a trapping network, its collection of trapspaces, and its collection of principal trapspaces.


\paragraph{The classification theorem}

Recall that $\Collections(n)$ denotes the set of all collections of subcubes of $\B^n$. Say a collection $\mathcal{J} \in \Collections( n )$ of subcubes is \Define{ideal} if it is the collection of trapspaces of a network. We denote the set of all ideal collections of subcubes of $\B^n$ by $\Collections^\Trapspaces(n)$. Accordingly, say a collection $\mathcal{Q}$ of subcubes is \Define{principal} if it is the collection of principal trapspaces of a network and denote the set of all principal collections of subcubes of $\B^n$ by $\Collections^\PrincipalTrapspaces(n)$. We shall give combinatorial descriptions of ideal and principal collections of subcubes in the sequel.


Define the mapping $F : \Collections(n) \to \Functions(n)$ as follows. Let $\mathcal{A} \in \Collections(n)$ be a collection of subcubes of $\B^n$. For any $x \in \B^n$, denote the intersection of all the subcubes in $\mathcal{A}$ that contain $x$ by
\[
    \mathcal{A}(x) := \bigcap \{ A \in \mathcal{A} : x \in A \},
\]
where $\mathcal{A}(x) = \B^n$ if the intersection is empty.  Then let $F( \mathcal{A} )$ be the network defined by
\[
    F( \mathcal{A} )(x) = \mathcal{A}(x) - x,
\]
or equivalently $\mathcal{A}(x) = [ x, F( \mathcal{A} )(x) ]$, for all $x \in \B^n$. Let $F_\PrincipalTrapspaces$ be the restriction of $F$ to $\Collections^\PrincipalTrapspaces(n)$ and $F_\Trapspaces$ be the restriction of $F$ to $\Collections^\Trapspaces(n)$.

Moreover, define the mappings $\lambda, \mu : \Collections( n ) \to \Collections( n )$ given by 
\[
    \lambda( \mathcal{A} ) = \left\{ \bigcup R, R \subseteq \mathcal{A} : \bigcup R \in \Subcubes( n ) \right\}
\]
and 
\[
    \mu( \mathcal{A} ) = \{ \mathcal{A}( x ) : x \in \B^n \}.
\]
Then let $\lambda_\PrincipalTrapspaces$ be the restriction of $\lambda$ to $\Collections^\PrincipalTrapspaces( n )$ and $\mu_\Trapspaces$ be the restriction of $\mu$ to $\Collections^\Trapspaces( n )$.

\begin{theorem}[Three-way equivalence for collections of trapspaces] \label{theorem:three-way_equivalence}
The diagram on Figure \ref{figure:three-way_equivalence} commutes. More concretely, the following hold.

\begin{enumerate}
    \item \label{item:Q-g_equivalence}
    For all $\mathcal{Q} \in \Collections^\PrincipalTrapspaces(n)$ and all $g \in \Trapping{\Functions}(n)$, we have
    \[
        \PrincipalTrapspaces( F_\PrincipalTrapspaces( \mathcal{Q} ) ) = \mathcal{Q}, \qquad F_\PrincipalTrapspaces( \PrincipalTrapspaces(g) ) = g.
    \]
    
    \item \label{item:J-g_equivalence}
    For all $\mathcal{J} \in \Collections^\Trapspaces(n)$ and all $g \in \Trapping{\Functions}(n)$, we have
    \[
        \Trapspaces( F_\Trapspaces( \mathcal{J} ) ) = \mathcal{J}, \qquad F_\Trapspaces( \Trapspaces(g) ) = g.
    \]
    
    \item \label{item:composition} 
    For all $\mathcal{Q} \in \Collections^\PrincipalTrapspaces( n )$ and all $\mathcal{J} \in \Collections^\Trapspaces( n )$, we have 
    \[
        \lambda_\PrincipalTrapspaces( \mathcal{Q} ) = \Trapspaces( F_\PrincipalTrapspaces( \mathcal{Q} ) ), \qquad \mu_\Trapspaces( \mathcal{J} ) = \PrincipalTrapspaces( F_\Trapspaces( \mathcal{J} ) ).
    \]

    \item \label{item:Q-J_equivalence}
    For all $\mathcal{Q} \in \Collections^\PrincipalTrapspaces( n )$ and all $\mathcal{J} \in \Collections^\Trapspaces( n )$, we have 
    \[
        \mu_\Trapspaces( \lambda_\PrincipalTrapspaces ( \mathcal{Q} ) ) = \mathcal{Q}, \qquad \lambda_\PrincipalTrapspaces( \mu_\Trapspaces( \mathcal{J} ) ) = \mathcal{J}.
    \]
\end{enumerate}

\end{theorem}

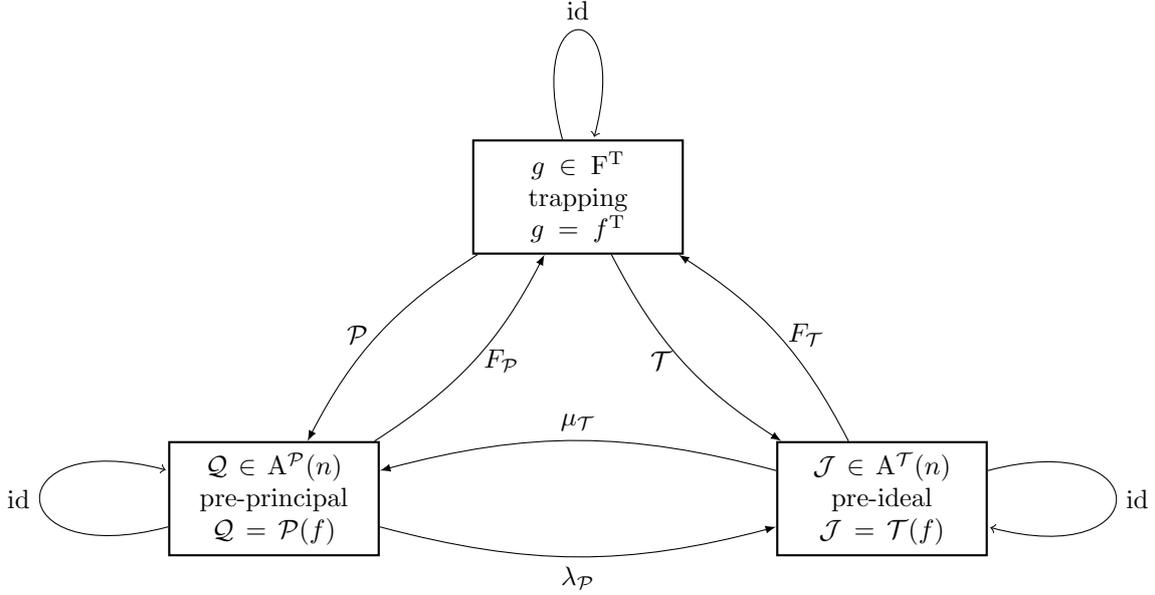
\begin{figure}
    \centering
\begin{tikzpicture}[scale = 4, font=\small]
    \node[draw,thick,minimum width=2cm,minimum height=1.5cm,  text width = 2.5cm, align=center] (g) at (1,1) {$g \in \Trapping{\Functions}$ \\ trapping\\ $g = \Trapping{f}$ };
    \node[draw,thick,minimum width=2cm,minimum height=1.5cm,  text width = 2.5cm, align=center] (Q) at (0,0) {$\mathcal{Q} \in \Collections^\PrincipalTrapspaces( n )$ pre-principal\\ $\mathcal{Q} = \PrincipalTrapspaces(f)$};
    \node[draw,thick,minimum width=2cm,minimum height=1.5cm,  text width = 2.5cm, align=center] (J) at (2,0) {$\mathcal{J} \in \Collections^\Trapspaces( n )$ pre-ideal\\ $\mathcal{J} = \Trapspaces(f)$};

    \draw[-latex] (g) to [bend right=15] node[left] {$\PrincipalTrapspaces$} (Q);
    \draw[-latex] (Q) to [bend right=15] node[right] {$F_\PrincipalTrapspaces$} (g);

    \draw[-latex] (g) to [bend right=15] node[left] {$\Trapspaces$} (J);
    \draw[-latex] (J) to [bend right=15] node[right] {$F_\Trapspaces$} (g);

    \draw[-latex] (Q) to [bend right=15] node[below] {$\lambda_\PrincipalTrapspaces$} (J);
    \draw[-latex] (J) to [bend right=15] node[above] {$\mu_\Trapspaces$} (Q);

    \path[-latex] (g) edge [loop above] node {$\id$} (g);
    \path[-latex] (Q) edge [loop left] node {$\id$} (Q);
    \path[-latex] (J) edge [loop right] node {$\id$} (J);

    
\end{tikzpicture}
    \caption{Three-way equivalence amongst collections of principal trapspaces, collections of trapspaces, and trapping closures.}
    \label{figure:three-way_equivalence}
\end{figure}

The rest of this subsection is devoted to the proof of Theorem \ref{theorem:three-way_equivalence}. We first prove item \ref{item:Q-g_equivalence}, by characterising the principal collections of subcubes.

Let $\mathcal{Q} \in \Collections(n)$ be a collection of subcubes of $\B^n$. We say $\mathcal{Q}$ is \Define{pre-principal} if 
\[
    \mathcal{Q} = \mu( \mathcal{Q} ) = \{ \mathcal{Q}(x) : x \in \B^n \}.
\]
Intuitively, $\mathcal{Q}$ is pre-principal if for any configuration $x$, there exists a smallest subcube in $\mathcal{Q}$ that contains $x$, and conversely for any subcube $A \in \mathcal{Q}$, there is a configuration $x$ for which $A$ is the smallest subcube that contains $x$.

\begin{lemma} \label{lemma:trapping_collection_of_subcubes}
A collection $\mathcal{Q}$ of subcubes of $\B^n$ is pre-principal if and only if the following hold:
\begin{enumerate}
    \item \label{item:bigcup(X)=Bn}
    $\bigcup \mathcal{Q} = \B^n$;

    \item \label{item:bigcup(R)=XcapY}
    for all $A, B \in \mathcal{Q}$, there exists $\mathcal{C} \subseteq \mathcal{Q}$ such that $\bigcup \mathcal{C} = A \cap B$;

    \item \label{item:bigcup(R)neX}
    for all $A \in \mathcal{Q}$ and $\mathcal{C} \subseteq \mathcal{Q}$, $\bigcup \mathcal{C} = A$ implies $A \in \mathcal{C}$.
\end{enumerate}
\end{lemma}

\begin{proof}
Suppose $\mathcal{Q}$ is pre-principal. We prove that it satisfies all three properties.
\begin{enumerate}
    \item 
    We have $x \in \mathcal{Q}(x)$ for all $x \in \B^n$, hence $\bigcup \mathcal{Q} = \B^n$.

    \item
    For all $A, B \in \mathcal{Q}$, $\bigcup\{ \mathcal{Q}(x) : x \in A \cap B \} = A \cap B$.

    \item
    Suppose $\mathcal{C} \subseteq \mathcal{Q}$ with $\bigcup \mathcal{C} = A \in \mathcal{Q}$ while $A \notin \mathcal{C}$. Then for all $x \in A$, there exists $C \in \mathcal{C}$ such that $\mathcal{Q}(x) \subseteq C \subset A$. Therefore, $A \notin \{ \mathcal{Q}(x) : x \in \B^n \}$, which contradicts the fact that $\mathcal{Q}$ is pre-principal.
\end{enumerate}

Conversely, let $\mathcal{Q}$ satisfy all three properties. We first prove that $\mathcal{Q}(x) \in \mathcal{Q}$ for all $x \in \B^n$. Let $x \in \B^n$ and consider the collection 
\[
    \mathcal{S} = \{ A \in \mathcal{Q} : x \in A;  \forall B \subset A, B \in \mathcal{Q}, x \notin B  \}
\]
of minimal subcubes in $\mathcal{Q}$ that contain $x$. By Property \ref{item:bigcup(X)=Bn}, $|\mathcal{S}| \ge 1$. If $|\mathcal{S}| \ge 2$, let $A,B \in \mathcal{S}$, then there exists $\mathcal{C} \subseteq \mathcal{Q}$ such that $\bigcup \mathcal{C} = A \cap B$. As such, there exists $C \in \mathcal{C}$ such that $x \in C$ while $C \subseteq A \cap B \subset A$, which contradicts the fact that $A \in \mathcal{S}$. Therefore, $|\mathcal{S}| = 1$, hence $\mathcal{S} = \{ \mathcal{Q}(x) \}$. 

We now prove that $A \in \{ \mathcal{Q}(x) : x \in \B^n \}$ for all $A  \in \mathcal{Q}$. Let $A \in \mathcal{Q}$, and suppose that $\mathcal{Q}(x) \subset A$ for all $x \in A$. Then $\mathcal{C} := \{ \mathcal{Q}(x) : x \in A \} \subset \mathcal{Q}$ satisfies $\bigcup \mathcal{C} = A$ while $A \notin \mathcal{C}$, which contradicts the third property.
\end{proof}

\begin{lemma} \label{lemma:principal_pre-principal}
A collection of subcubes is principal if and only if it is pre-principal. For all $\mathcal{Q} \in \Collections^\PrincipalTrapspaces(n)$ and all $g \in \Trapping{\Functions}(n)$, we have
\[
    \PrincipalTrapspaces( F_\PrincipalTrapspaces( \mathcal{Q} ) ) = \mathcal{Q}, \qquad F_\PrincipalTrapspaces( \PrincipalTrapspaces(g) ) = g.
\]
\end{lemma}

\begin{proof}
Firstly, we prove that $F_\PrincipalTrapspaces( \mathcal{Q} )$ is trapping. Placing ourselves in the graph $\Gamma = \GeneralAsynchronous( F_\PrincipalTrapspaces( \mathcal{Q} ) )$, if $y \in N^{out} ( \Gamma; x ) = \mathcal{Q}(x)$, then $N^{out}( \Gamma; y ) = \mathcal{Q}(y) \subseteq \mathcal{Q}(x) = N^{out}( \Gamma; x )$, and hence $\Gamma$ is transitive.

Secondly, we prove that $\PrincipalTrapspaces(g)$ is pre-principal by verifying that it satisfies the three properties of Lemma \ref{lemma:trapping_collection_of_subcubes}. First, since $x \in T_g(x)$ for all $x \in \B^n$, we have $\bigcup \PrincipalTrapspaces(g)  = \B^n$. Second, for all $A, B \in \PrincipalTrapspaces(g)$, the collection $\mathcal{C} = \{ T_g(x) : x \in A \cap B \}$ satisfies $\mathcal{C} \subseteq \PrincipalTrapspaces(g)$ and $\bigcup \mathcal{C} = A \cap B$. Third, if there exists $x \in \B^n$ and $\mathcal{C} \subseteq \PrincipalTrapspaces(g) \setminus \{ T_g(x) \}$ such that $\bigcup \mathcal{C} = T_g(x)$, then there exists $C \in \mathcal{C}$ such that $x \in C \subset T_g(x)$ and hence $T_g(x) \subseteq C \subset T_g(x)$, which is the desired contradiction.

Since any principal collection of subcubes is of the form $\PrincipalTrapspaces(g)$ for some trapping network $g$, we have just shown that any principal collection of subcubes is pre-principal.

Thirdly, we prove that $\PrincipalTrapspaces( F_\PrincipalTrapspaces( \mathcal{Q} ) ) = \mathcal{Q}$ for any pre-principal collection $\mathcal{Q}$ of subcubes of $\B^n$. Let $g = F( \mathcal{Q} )$. Since $g$ is trapping, we have for all $x \in \B^n$
\[
    T_g(x) = N^{out}( \GeneralAsynchronous(g); x ) = \mathcal{Q}(x).
\]
Hence $\PrincipalTrapspaces(g) = \{ \mathcal{Q}(x) : x \in \B^n \} = \mathcal{Q}$ since $\mathcal{Q}$ is pre-principal.

We have just shown that any pre-principal collection of subcubes is principal. Together with the previous item, we have proved that a collection of subcubes is principal if and only if it is pre-principal.

Fourthly, we prove that $F_\PrincipalTrapspaces( \PrincipalTrapspaces(g) ) = g$ for all $g \in \Trapping{\Functions}(n)$. Let $\mathcal{Q} = \PrincipalTrapspaces(g)$. For all $x \in \B^n$, we have
\[
    T_g(x) = \bigcap\{ A \in \PrincipalTrapspaces(g) : x \in A \} = \mathcal{Q}(x)
\]
hence $\mathcal{Q}(x) = [x, g(x)]$ (since $g$ is trapping). On the other hand, by definition $\mathcal{Q}(x) = [ x, F_\PrincipalTrapspaces( \mathcal{Q} ) (x) ]$, thus $g = F_\PrincipalTrapspaces( \mathcal{Q} )$.
\end{proof}

We now prove item \ref{item:J-g_equivalence}, by characterising ideal collections of subcubes. We say that a collection $\mathcal{J}$ of subcubes is \Define{pre-ideal} if it satisfies the following three properties:
\begin{enumerate}
    \item \label{item:pre-ideal1}
    $\B^n \in \mathcal{J}$;

    \item \label{item:pre-ideal2}
    $\mathcal{J}$ is closed under intersection, i.e. if $A, B \in \mathcal{J}$ and $A \cap B \ne \emptyset$, then $A \cap B \in \mathcal{J}$;

    \item \label{item:pre-ideal3}
    for any subcollection $\mathcal{R} \subseteq \mathcal{J}$, if $R = \bigcup \mathcal{R} \in \Subcubes(n)$, then $R \in \mathcal{J}$.
\end{enumerate}
Intuitively, a pre-ideal collection of subcubes is closed under arbitrary unions and intersections, so long as those unions and intersections are actual subcubes. We note that property \ref{item:pre-ideal3} above is equivalent to $\mathcal{J} = \lambda( \mathcal{J} )$.

\begin{lemma} \label{lemma:ideal_pre-ideal}
A collection of subcubes is ideal if and only if it is pre-ideal. For all $\mathcal{J} \in \Collections^\Trapspaces(n)$ and all $g \in \Trapping{\Functions}(n)$, we have
\[
    \Trapspaces( F_\Trapspaces( \mathcal{J} ) ) = \mathcal{J}, \qquad F_\Trapspaces( \Trapspaces(g) ) = g.
\]
\end{lemma}

\begin{proof}
The proof uses a similar structure to that of Lemma \ref{lemma:principal_pre-principal}.

Firstly, $F_\Trapspaces( \mathcal{J} )$ is trapping. (The proof is the same as its counterpart for $F_\PrincipalTrapspaces( \mathcal{Q} )$.)

Secondly, we prove that $\Trapspaces( g )$ is pre-ideal. First, $\B^n$ is a trapspace of $g$, hence $\B^n \in \mathcal{J}$. Second, if $A$ and $B$ are trapspaces with non-empty intersection, let $x \in A \cap B$, then $g( x ) \in A \cap B$, hence $A \cap B$ is also a trapspace. Third, if $R = \bigcup \mathcal{R}$ for some $\mathcal{R} \subseteq \mathcal{J}$, then for any $x \in R$, there exists $A \in \mathcal{J}$ such that $x \in A$ and hence $g( x ) \in A \subseteq R$, thus $R$ is also a trapspace.

Since any ideal collection of subcubes is of the form $\Trapspaces(g)$ for some trapping network $g$, we have just shown that any ideal collection of subcubes is pre-ideal.

Thirdly, we prove that $\Trapspaces( F_\Trapspaces( \mathcal{J} ) ) = \mathcal{J}$ for any pre-ideal collection $\mathcal{J}$ of subcubes of $\B^n$. Let $g = F_\Trapspaces( \mathcal{J} )$ so that $[ x, g(x) ] = \mathcal{J}(x) \in \mathcal{J}$ for all $x \in \B^n$. Suppose $A \in \mathcal{J}$, then for all $x \in A$, $g(x) \in \mathcal{J}(x) \subseteq A$, hence $A \in \Trapspaces( g )$. Conversely, suppose $B \in \Trapspaces( g )$, then $B = \bigcup \{ \mathcal{J}(x) : x \in B \}$, whence $B \in \mathcal{J}$ since $\mathcal{J}$ is pre-ideal.

We have just shown that any pre-ideal collection of subcubes is ideal. Together with the previous item, we have proved that a collection of subcubes is ideal if and only if it is pre-ideal.

Fourthly, we prove that $F_\Trapspaces( \Trapspaces( g ) ) = g$ for all $g \in \Trapping{\Functions}(n)$. Let $\mathcal{Q} = \Trapspaces( g )$, so that $\mathcal{Q}(x) = [ x, g(x) ]$ for all $x \in \B^n$ since $g$ is trapping. Thus $g = F_\Trapspaces( \mathcal{Q} )$ by definition of $F$.
\end{proof}

We now prove item \ref{item:composition}.

\begin{lemma}
For all $\mathcal{Q} \in \Collections^\PrincipalTrapspaces( n )$ and all $\mathcal{J} \in \Collections^\Trapspaces( n )$, we have 
\[
    \lambda_\PrincipalTrapspaces( \mathcal{Q} ) = \Trapspaces( F_\PrincipalTrapspaces( \mathcal{Q} ) ), \qquad \mu_\Trapspaces( \mathcal{J} ) = \PrincipalTrapspaces( F_\Trapspaces( \mathcal{J} ) ).
\]
\end{lemma}

\begin{proof}
We first prove that $\lambda_\PrincipalTrapspaces( \mathcal{Q} ) = \Trapspaces( F_\PrincipalTrapspaces( \mathcal{Q} ) )$. Let $g = F_\PrincipalTrapspaces( \mathcal{Q} )$. By Lemma \ref{lemma:principal_pre-principal}, $g$ is trapping with collection of principal trapspaces $\PrincipalTrapspaces( g ) = \mathcal{Q}$. Therefore its collection of trapspaces is given by any union of elements of $\mathcal{Q}$ that form a subcube, i.e. $\Trapspaces( g ) = \lambda_\PrincipalTrapspaces( \mathcal{Q})$.

We now prove that $\mu_\Trapspaces( \mathcal{J} ) = \PrincipalTrapspaces( F_\Trapspaces( \mathcal{J} ) )$. Again, let $g = F_\Trapspaces( \mathcal{J} )$; then $g$ is trapping with collection of trapspaces $\Trapspaces( g ) = \mathcal{J}$. Therefore, $\PrincipalTrapspaces( g ) = \{ \bigcap_{x \in A, A \in \mathcal{J}} A : x \in \B^n \} = \mu_\Trapspaces( \mathcal{J} )$.
\end{proof}

Item \ref{item:Q-J_equivalence} immediately follows.

\begin{corollary}
For all $\mathcal{Q} \in \Collections^\PrincipalTrapspaces( n )$ and all $\mathcal{J} \in \Collections^\Trapspaces( n )$, we have 
\[
    \mu_\Trapspaces( \lambda_\PrincipalTrapspaces ( \mathcal{Q} ) ) = \mathcal{Q}, \qquad \lambda_\PrincipalTrapspaces( \mu_\Trapspaces( \mathcal{J} ) ) = \mathcal{J}.
\]
\end{corollary}

\subsection{Minimal trapspaces and min-trapping networks} \label{subsection:minimal_trapspaces}

Part of the theory built for principal trapspaces and trapping networks in the prequel of this section can be adapted to study minimal trapspaces instead. 

We first note that the collection $\MinimalTrapspaces(f)$ of minimal trapspaces of $f$ does not determine the collection $\PrincipalTrapspaces(f)$ of principal trapspaces of $f$. For instance, consider the following two networks, given by their respective asynchronous graphs below. They have the same collection of minimal trapspaces, namely the two fixed points, but the line $\{ x_1 = 1 \}$ is a principal trapspace of the first network but not of the second.

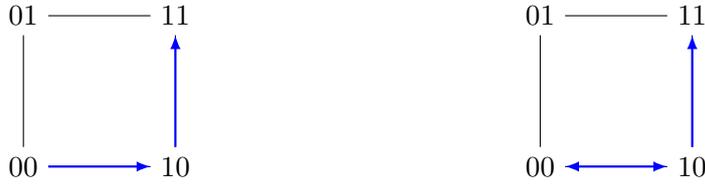
\begin{center}
    
\begin{tikzpicture}
    \node (00) at (0,0) {$00$};
    \node (01) at (0,2) {$01$};
    \node (10) at (2,0) {$10$};
    \node (11) at (2,2) {$11$};

    \path[draw] (00) -- (01) -- (11)
    (00) -- (10) -- (11);
    
    \draw[very thick,-latex, blue] (00) -- (10);
    \draw[very thick,-latex, blue] (10) -- (11);

    \begin{scope}[xshift = 5cm]
        
    \node (00) at (0,0) {$00$};
    \node (01) at (0,2) {$01$};
    \node (10) at (2,0) {$10$};
    \node (11) at (2,2) {$11$};

    \path[draw] (00) -- (01) -- (11)
    (00) -- (10) -- (11);
    
    \draw[very thick,latex-latex, blue] (00) -- (10);
    \draw[very thick,-latex, blue] (10) -- (11);
    \end{scope}

\end{tikzpicture}
\end{center}

Since the trapping closure of $f$ satisfies $\Trapping{f} = F( \PrincipalTrapspaces(f) )$, we define the \Define{min-trapping extension} of $f$ by
\[
    \MinimalTrapping{f} = F( \MinimalTrapspaces(f) ).
\]
More explicitly, say $x$ is a \Define{min-trapspace configuration} of $f$ if it belongs to a minimal trapspace of $f$ and denote the set of min-trapspace configurations of $f$ by $M(f)$. Then the min-trapping extension of $f$ is given by
\[
    \MinimalTrapping{f}( x ) = \begin{cases}
        T_f( x ) - x & \text{if } x \in M(f) \\
        \neg x & \text{otherwise}.
    \end{cases}
\]
The min-trapping extension is not a closure operator, as below we display two networks $f$ and $g$ such that $f \sqsubseteq g$ but $\MinimalTrapping{f} \not\sqsubseteq \MinimalTrapping{g}$.

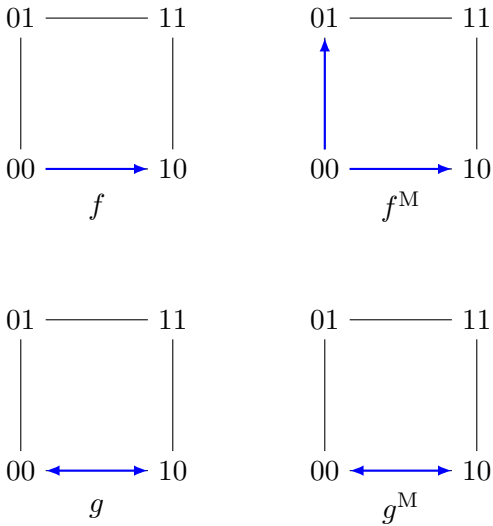
\begin{center}
\begin{tikzpicture}

\begin{scope}[xshift=0cm, yshift=0cm]
    \node (f) at (1, -0.5) {$f$};

    \node (00) at (0,0) {$00$};
    \node (01) at (0,2) {$01$};
    \node (10) at (2,0) {$10$};
    \node (11) at (2,2) {$11$};

    \path[draw] (00) -- (01) -- (11)
    (00) -- (10) -- (11);
    
    
    \draw[very thick, -latex, blue] (00) -- (10);

\end{scope} 

\begin{scope}[xshift=5cm, yshift=0cm]
    \node (fM) at (1, -0.5) {$\MinimalTrapping{f}$};

    \node (00) at (0,0) {$00$};
    \node (01) at (0,2) {$01$};
    \node (10) at (2,0) {$10$};
    \node (11) at (2,2) {$11$};

    \path[draw] (00) -- (01) -- (11)
    (00) -- (10) -- (11);
    
    
    \draw[very thick,-latex, blue] (00) -- (10);
    \draw[very thick,-latex, blue] (00) -- (01);  

\end{scope} 

\begin{scope}[xshift=0cm, yshift=-4cm]
    \node (g) at (1, -0.5) {$g$};

    \node (00) at (0,0) {$00$};
    \node (01) at (0,2) {$01$};
    \node (10) at (2,0) {$10$};
    \node (11) at (2,2) {$11$};

    \path[draw] (00) -- (01) -- (11)
    (00) -- (10) -- (11);
    
    
    \draw[very thick,latex-latex, blue] (00) -- (10);
    
\end{scope} 

\begin{scope}[xshift=5cm, yshift=-4cm]
    \node (gM) at (1, -0.5) {$\MinimalTrapping{g}$};

    \node (00) at (0,0) {$00$};
    \node (01) at (0,2) {$01$};
    \node (10) at (2,0) {$10$};
    \node (11) at (2,2) {$11$};

    \path[draw] (00) -- (01) -- (11)
    (00) -- (10) -- (11);
    
    
    \draw[very thick,latex-latex, blue] (00) -- (10);
    
\end{scope} 

\end{tikzpicture}
\end{center}

The min-trapping extension does satisfy all the other properties that we expect. The proof of Lemma \ref{lemma:properties_M} is straightforward and hence omitted.

\begin{lemma} \label{lemma:properties_M}
For any $f \in \Functions( n )$, the following hold.
\begin{align}
    f &\sqsubseteq \MinimalTrapping{f}, \\
    \MinimalTrapping{( \MinimalTrapping{f} )} &= \MinimalTrapping{f}, \\
    \MinimalTrapspaces( \MinimalTrapping{f} ) &= \MinimalTrapspaces( f ).
\end{align}
\end{lemma}

Say a network $g$ is \Define{min-trapping} if $g = \MinimalTrapping{g}$ and denote the set of min-trapping networks in $\Functions(n)$ by $\MinimalTrapping{ \Functions }( n )$. The min-trapping extension of $f$ is a trapping network that satisfies $\Trapping{f} \sqsubseteq \MinimalTrapping{f}$, as such we have
\[
    \MinimalTrapping{f} = \Trapping{ ( \MinimalTrapping{f} ) } = \MinimalTrapping{ ( \Trapping{f} ) } = \MinimalTrapping{ ( \MinimalTrapping{f} ) }.
\]

We now give the analogue of Theorem \ref{theorem:T(f)=T(g)} for min-trapping extensions.

\begin{theorem}  \label{theorem:min-trapping_closure}
The following are equivalent for $f, g \in \Functions(n)$:
\begin{enumerate}
    \item \label{item:M(f)=M(g)}
    $\MinimalTrapspaces(f) = \MinimalTrapspaces(g)$;

    \item \label{item:Tf(x)=Tg(x)_min_equal}
    $M(f) = M(g)$ and $T_f( x ) = T_g( x )$ for all $x \in M(f)$;

    \item \label{item:Tf(x)=Tg(x)_min}
   $T_f( x ) = T_g( x )$ for all $x \in M(f) \cup M(g)$;

    \item \label{item:fM=gM}
    $\MinimalTrapping{f} = \MinimalTrapping{g}$.
\end{enumerate}
\end{theorem}

\begin{proof}

$\ref{item:M(f)=M(g)} \implies \ref{item:Tf(x)=Tg(x)_min_equal}$. We have $M(f) = \{ x \in A : A \in \MinimalTrapspaces(f) \} = \{ x \in A : A \in \MinimalTrapspaces(g) \} = M(g)$. Now, for any $x \in M(f)$, $x$ belongs to a unique minimal trapspace of $f$, namely $T_f(x)$; since $M(f) = M(g)$, $x$ belongs to a unique minimal trapspace of $g$, namely $T_g(x)$. Therefore, $T_f(x) = T_g(x)$.

$\ref{item:Tf(x)=Tg(x)_min_equal} \implies \ref{item:Tf(x)=Tg(x)_min}$. Trivial.

$\ref{item:Tf(x)=Tg(x)_min} \implies \ref{item:Tf(x)=Tg(x)_min_equal}$. For the sake of contradiction, suppose that $x \in M(f) \setminus M(g)$, so that $T_f(x) = T_g(x)$. Let $y \in T_g(x) \cap M(g)$, then $T_f(y) = T_g(y) \subset T_g(x) = T_f(x)$, and hence $x \notin M(f)$, which is the desired contradiction. Thus $M(f) \subseteq M(g)$, and by symmetry we obtain $M(f) = M(g)$.

$\ref{item:Tf(x)=Tg(x)_min_equal} \implies \ref{item:fM=gM}$. By definition of the min-trapping extension.

$\ref{item:fM=gM} \implies \ref{item:M(f)=M(g)}$. We have $\MinimalTrapspaces( f ) = \MinimalTrapspaces( \MinimalTrapping{f} ) = \MinimalTrapspaces( \MinimalTrapping{g} ) = \MinimalTrapspaces( g )$.

\end{proof}

Say a collection of subcubes $\mathcal{N} \in \Collections(n)$ is \Define{min-ideal} if all its elements are disjoint: $A \cap B = \emptyset$ for all $A \ne B \in \mathcal{N}$. We denote the set of all min-ideal collections of subcubes of $\B^n$ by $\Collections^\MinimalTrapspaces(n)$, and the restriction of the mapping $F$ to $\Collections^\MinimalTrapspaces(n)$ as $F_\MinimalTrapspaces$.

\begin{theorem} \label{theorem:two-way_equivalence}
The set of min-trapping networks is in bijection with the set of min-ideal collections of subcubes. More precisely, for all min-ideal collections of subcubes $\mathcal{N} \in \Collections^\MinimalTrapspaces(n)$ and all min-trapping networks $g \in \MinimalTrapping{ \Functions }( n )$, we have
\[
    \MinimalTrapspaces( F_\MinimalTrapspaces( \mathcal{N} ) ) = \mathcal{N}, \qquad F_\MinimalTrapspaces( \MinimalTrapspaces( g ) ) = g.
\]
\end{theorem}

\begin{proof}
We first prove that $\MinimalTrapspaces( F_\MinimalTrapspaces( \mathcal{N} ) ) = \mathcal{N}$ for any min-ideal collection $\mathcal{N}$. Let $N = \bigcup_{A \in \mathcal{N}} A$ denote the content of $\mathcal{N}$. For any $x \in \B^n$, we have
\[
    \mathcal{N}(x) = \begin{cases}
        A & \text{if } x \in A, A \in \mathcal{N} \\
        \B^n & \text{if } x \notin N.
    \end{cases}
\]
Therefore, $g = F_\MinimalTrapspaces( \mathcal{N} )$ is given by
\[
    g( x ) = \begin{cases}
        A - x & \text{if } x \in A, A \in \mathcal{N} \\
        \neg x & \text{if } x \notin N.
    \end{cases}
\]
We obtain that $\MinimalTrapspaces( g ) = \mathcal{N}$.

We now prove that $F_\MinimalTrapspaces( \MinimalTrapspaces( g ) ) = g$ for any min-trapping network $g$. Let $\mathcal{N} = \MinimalTrapspaces( g )$, so that $\mathcal{N} = \{ T_g( x ) : x \in M(g) \}$ is min-ideal. Thus,
\[
    \mathcal{N}(x) = \begin{cases}
        T_g(x) & \text{if } x \in M(g) \\
        \B^n & \text{otherwise}
    \end{cases}
\]
and
\[
    F_\MinimalTrapspaces( \mathcal{N} )(x) = \begin{cases}
        T_g(x) - x & \text{if } x \in M(g) \\
        \neg x & \text{otherwise}
    \end{cases}
     = \MinimalTrapping{g}( x ) = g(x).
\]

\end{proof}

\section{Commutative networks} \label{section:commutative}

\subsection{Review of commutative networks} \label{subsection:review_commutative}

\paragraph{Definitions}

In \cite{BGT20}, Bridoux et al. introduce commutative networks, where asynchronous updates can be performed in any order without altering the result. This subsection is devoted to a review of some of the results in \cite{BGT20}. Foremost, the authors of \cite{BGT20} are interested in possibly infinite networks, and hence distinguish between so-called locally and globally commutative networks. However, these two concepts coincide for the finite Boolean networks we study in this paper. As such, we say a network $f \in \Functions(n)$ is \Define{commutative} if it satisfies the following three equivalent properties:
\begin{enumerate}
    \item 
    for all $i, j \in [n]$, $f^{ (i, j) } = f^{ (j, i) }$;

    \item 
    for all $S, T \subseteq [n]$, $f^{ (S, T) } = f^{ (T, S)}$;

    \item 
    for all $S, T \subseteq [n]$ with $S \cap T = \emptyset$, $f^{ (S, T) } = f^{ (S \cup T) }$.
\end{enumerate}

\paragraph{Properties}

Commutative networks have some strong structural properties. Intuitively, they all stem from the key property that any update can be ``serialised'', i.e. for any $S = \{ s_1, \dots, s_k \} \subseteq [n]$, we have
\[
    f^{(S)} = f^{(s_1, \dots, s_k)} = f^{(s_k)} \circ \dots \circ f^{(s_1)}.
\]
For instance, for any property $P$ of Boolean networks, we say $f$ is \Define{locally $P$} if $f^{(i)}$ satisfies $P$ for all $i \in [n]$ and \Define{globally $P$} if $f^{(S)}$ satisfies $P$ for all $S \subseteq [n]$.
As such, a network $f \in \Functions( n )$ is \Define{bijective} if $f$ is a bijection (i.e., a permutation of $\B^n$); $f$ is \Define{locally bijective} if $f^{(i)}$ is bijective for all $i \in [n]$; $f$ is \Define{globally bijective} if $f^{ (S) }$ is a bijection for all $S \subseteq [n]$. If $f$ is a locally bijective commutative network, then for all $S \subseteq [n]$, $f^{(S)} = f^{(s_k)} \circ \dots \circ f^{(s_1)}$ is also bijective, i.e. $f$ is globally bijective. Bridoux et al. go further, and show that the following are equivalent for a commutative network:
\begin{enumerate}
    \item $f$ is bijective;

    \item $f$ is locally bijective;

    \item $f$ is globally bijective.
\end{enumerate}
Moreover, recall that $f$ is \Define{idempotent} if $f^2 = f$. The following are also equivalent for a commutative network:
\begin{enumerate}
    \item $f$ is idempotent;

    \item $f$ is locally idempotent;

    \item $f$ is globally idempotent.
\end{enumerate}

Commutative networks have heavily constrained dynamics. Any function $\phi : \B \to \B$ has transient length at most $1$ and period at most $2$, and hence $\phi^3 = \phi$.  We then call a network $f$ \Define{dynamically local} if $f^3 = f$. Clearly, for any $f \in \Functions(n)$ and any $i \in [n]$ the update $f^{(i)}$ is dynamically local. In fact, commutative networks are also dynamically local, i.e. they have transient length at most $1$ and period at most $2$. Moreover, a network $f \in \Functions( n )$ is \Define{involutive} if $f$ is an involution, i.e. $f^2 = \id$. Any involutive network is bijective, and as seen above any bijective commutative network is involutive. 

\paragraph{Classification of asynchronous graphs}

The pinnacle of the study of commutative Boolean networks in \cite{BGT20} is a full classification of the asynchronous graphs of commutative networks. We give an overview of the classification below, and guide the reader to \cite{BGT20} for more detail. The classification takes three main steps: arrangement, arrangement network, and union of arrangement networks.

First, a family of subcubes $\mathcal{X}$ is called an \Define{arrangement} if $Y := \bigcap_{X \in \mathcal{X}} X$ is a non-empty subcube. 
We denote the \Define{content} of $X$ by $\hat{\mathcal{X}} := \bigcup_{X \in \mathcal{X}} X$.
We say $i \in [n]$ is a \Define{free dimension} of $\mathcal{X}$ if, for any $x \in \hat{\mathcal{X}}$, $y = (\neg x_i, x_{-i}) \in \hat{\mathcal{X}}$ as well. 
Six examples of arrangements are displayed below.

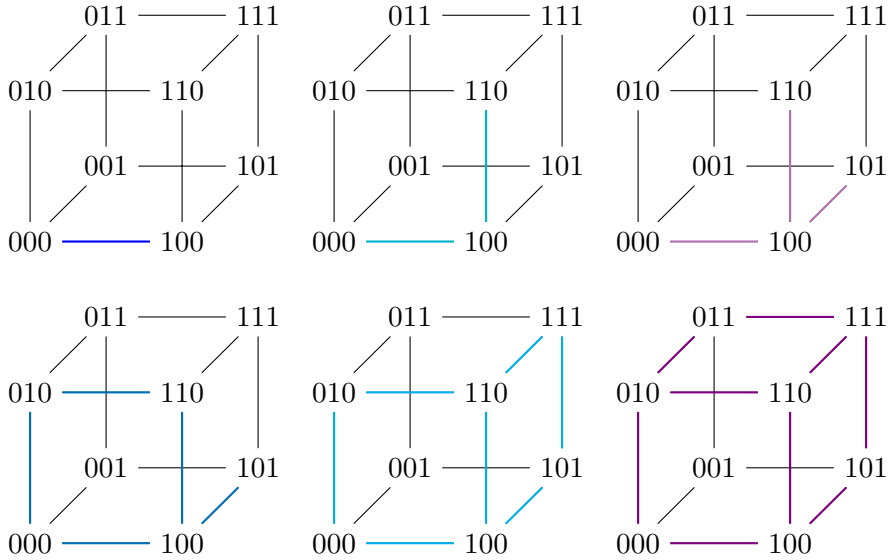
\begin{center}
\begin{tikzpicture}
\begin{scope}[xshift=0cm]
    \node (000) at (0,0) {$000$};
    \node (001) at (1,1) {$001$};
    \node (010) at (0,2) {$010$};
    \node (011) at (1,3) {$011$};
    \node (100) at (2,0) {$100$};
    \node (101) at (3,1) {$101$};
    \node (110) at (2,2) {$110$};
    \node (111) at (3,3) {$111$};

    \path[draw] (000) -- (001) -- (011) -- (111)
    (000) -- (010) -- (110) -- (111)
    (000) -- (100) -- (101) -- (111)
    (001) -- (101)
    (010) -- (011)
    (100) -- (110);
    
    \path[draw, very thick, blue] (000) -- (100);

\end{scope}

\begin{scope}[yshift=-5cm, xshift=0cm]
    \node (000) at (0,0) {$000$};
    \node (001) at (1,1) {$001$};
    \node (010) at (0,2) {$010$};
    \node (011) at (1,3) {$011$};
    \node (100) at (2,0) {$100$};
    \node (101) at (3,1) {$101$};
    \node (110) at (2,2) {$110$};
    \node (111) at (3,3) {$111$};

    \path[draw] (000) -- (001) -- (011) -- (111)
    (000) -- (010) -- (110) -- (111)
    (000) -- (100) -- (101) -- (111)
    (001) -- (101)
    (010) -- (011)
    (100) -- (110);
    
    \path[draw, very thick, NavyBlue] (000) -- (100) -- (110) -- (010) -- (000);
    \path[draw, very thick, NavyBlue] (100) -- (101);

\end{scope}

\begin{scope}[xshift=5cm]
    \node (000) at (0,0) {$000$};
    \node (001) at (1,1) {$001$};
    \node (010) at (0,2) {$010$};
    \node (011) at (1,3) {$011$};
    \node (100) at (2,0) {$100$};
    \node (101) at (3,1) {$101$};
    \node (110) at (2,2) {$110$};
    \node (111) at (3,3) {$111$};

    \path[draw] (000) -- (001) -- (011) -- (111)
    (000) -- (010) -- (110) -- (111)
    (000) -- (100) -- (101) -- (111)
    (001) -- (101)
    (010) -- (011)
    (100) -- (110);
    
    \path[draw, very thick, Turquoise] (000) -- (100) -- (110);

\end{scope}

\begin{scope}[yshift=-5cm, xshift=5cm]
    
    \node (000) at (0,0) {$000$};
    \node (001) at (1,1) {$001$};
    \node (010) at (0,2) {$010$};
    \node (011) at (1,3) {$011$};
    \node (100) at (2,0) {$100$};
    \node (101) at (3,1) {$101$};
    \node (110) at (2,2) {$110$};
    \node (111) at (3,3) {$111$};

    \path[draw] (000) -- (001) -- (011) -- (111)
    (000) -- (010) -- (110) -- (111)
    (000) -- (100) -- (101) -- (111)
    (001) -- (101)
    (010) -- (011)
    (100) -- (110);
    
    \path[draw, very thick, cyan] (000) -- (100) -- (110) -- (010) -- (000);
    \path[draw, very thick, cyan] (100) -- (101) -- (111) -- (110);
\end{scope}

\begin{scope}[xshift=10cm]
    \node (000) at (0,0) {$000$};
    \node (001) at (1,1) {$001$};
    \node (010) at (0,2) {$010$};
    \node (011) at (1,3) {$011$};
    \node (100) at (2,0) {$100$};
    \node (101) at (3,1) {$101$};
    \node (110) at (2,2) {$110$};
    \node (111) at (3,3) {$111$};

    \path[draw] (000) -- (001) -- (011) -- (111)
    (000) -- (010) -- (110) -- (111)
    (000) -- (100) -- (101) -- (111)
    (001) -- (101)
    (010) -- (011)
    (100) -- (110);
    
    \path[draw, very thick, Orchid] (000) -- (100) -- (110);
    \path[draw, very thick, Orchid] (100) -- (101);
\end{scope}

\begin{scope}[yshift=-5cm, xshift=10cm]
    
    \node (000) at (0,0) {$000$};
    \node (001) at (1,1) {$001$};
    \node (010) at (0,2) {$010$};
    \node (011) at (1,3) {$011$};
    \node (100) at (2,0) {$100$};
    \node (101) at (3,1) {$101$};
    \node (110) at (2,2) {$110$};
    \node (111) at (3,3) {$111$};

    \path[draw] (000) -- (001) -- (011) -- (111)
    (000) -- (010) -- (110) -- (111)
    (000) -- (100) -- (101) -- (111)
    (001) -- (101)
    (010) -- (011)
    (100) -- (110);
    
    \path[draw, very thick, violet] (000) -- (100) -- (110) -- (010) -- (000);
    \path[draw, very thick, violet] (100) -- (101) -- (111) -- (110);
    \path[draw, very thick, violet] (010) -- (011) -- (111);
\end{scope}
\end{tikzpicture}
\end{center}

Second, an \Define{arrangement network} intuitively only works on $\hat{\mathcal{X}}$, converges towards $Y$, and is uniform on any free dimension of $\mathcal{X}$. More formally, an arrangement network satisfies: $f(x) = x$ if $x \notin \hat{\mathcal{X}}$, $f(x) \in Y$ for all $x \in \hat{\mathcal{X}}$, and $f_i(x) = f_i(y)$ for all $x,y \in \hat{\mathcal{X}}$ with $x_i = y_i$. Arrangement networks are clearly commutative.
There are three possible arrangement networks for the arrangement in the bottom row, centre column above. They are displayed below.

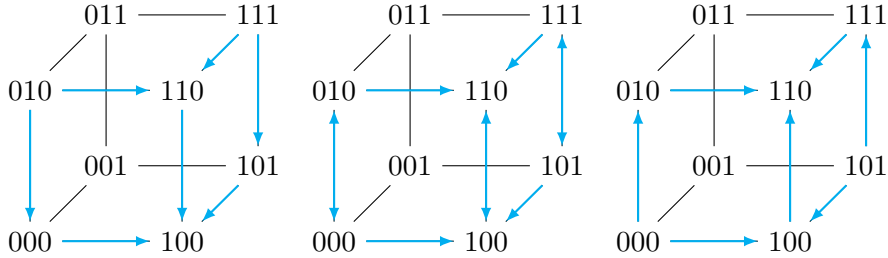
\begin{center}
\begin{tikzpicture}

\begin{scope}[xshift=0cm]
    \node (000) at (0,0) {$000$};
    \node (001) at (1,1) {$001$};
    \node (010) at (0,2) {$010$};
    \node (011) at (1,3) {$011$};
    \node (100) at (2,0) {$100$};
    \node (101) at (3,1) {$101$};
    \node (110) at (2,2) {$110$};
    \node (111) at (3,3) {$111$};

    \path[draw] (000) -- (001) -- (011) -- (111)
    (000) -- (010) -- (110) -- (111)
    (000) -- (100) -- (101) -- (111)
    (001) -- (101)
    (010) -- (011)
    (100) -- (110);
    
    \draw[very thick, cyan, -latex] (010) -- (110);
    \draw[very thick, cyan, -latex] (000) -- (100);
    \draw[very thick, cyan, -latex] (111) -- (110);
    \draw[very thick, cyan, -latex] (101) -- (100);

    \draw[very thick, cyan, -latex] (010) -- (000);
    \draw[very thick, cyan, -latex] (110) -- (100);
    \draw[very thick, cyan, -latex] (111) -- (101);
\end{scope}

\begin{scope}[xshift=5cm]
    \node (000) at (0,0) {$000$};
    \node (001) at (1,1) {$001$};
    \node (010) at (0,2) {$010$};
    \node (011) at (1,3) {$011$};
    \node (100) at (2,0) {$100$};
    \node (101) at (3,1) {$101$};
    \node (110) at (2,2) {$110$};
    \node (111) at (3,3) {$111$};

    \path[draw] (000) -- (001) -- (011) -- (111)
    (000) -- (010) -- (110) -- (111)
    (000) -- (100) -- (101) -- (111)
    (001) -- (101)
    (010) -- (011)
    (100) -- (110);
    
    \draw[very thick, cyan, -latex] (010) -- (110);
    \draw[very thick, cyan, -latex] (000) -- (100);
    \draw[very thick, cyan, -latex] (111) -- (110);
    \draw[very thick, cyan, -latex] (101) -- (100);

    \draw[very thick, cyan, latex-latex] (010) -- (000);
    \draw[very thick, cyan, latex-latex] (110) -- (100);
    \draw[very thick, cyan, latex-latex] (111) -- (101);
\end{scope}

\begin{scope}[xshift=10cm]
    \node (000) at (0,0) {$000$};
    \node (001) at (1,1) {$001$};
    \node (010) at (0,2) {$010$};
    \node (011) at (1,3) {$011$};
    \node (100) at (2,0) {$100$};
    \node (101) at (3,1) {$101$};
    \node (110) at (2,2) {$110$};
    \node (111) at (3,3) {$111$};

    \path[draw] (000) -- (001) -- (011) -- (111)
    (000) -- (010) -- (110) -- (111)
    (000) -- (100) -- (101) -- (111)
    (001) -- (101)
    (010) -- (011)
    (100) -- (110);
    
    \draw[very thick, cyan, -latex] (010) -- (110);
    \draw[very thick, cyan, -latex] (000) -- (100);
    \draw[very thick, cyan, -latex] (111) -- (110);
    \draw[very thick, cyan, -latex] (101) -- (100);

    \draw[very thick, cyan, -latex] (000) -- (010);
    \draw[very thick, cyan, -latex] (100) -- (110);
    \draw[very thick, cyan, -latex] (101) -- (111);

\end{scope}

\end{tikzpicture}
\end{center}

Third, one can take the \Define{union of arrangement networks}, provided their respective arrangement contents are disjoint, by taking the union of their asynchronous graphs. Three examples are displayed below.

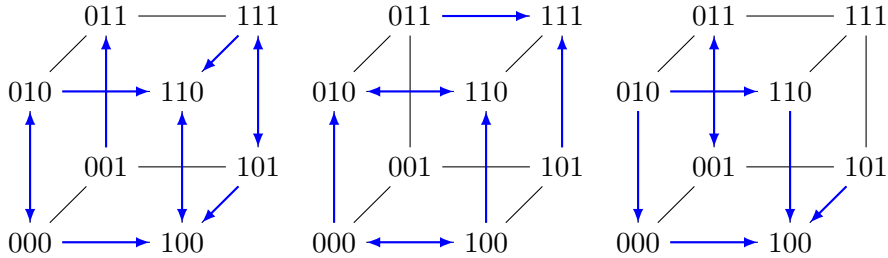
\begin{center}
    
\begin{tikzpicture}

\begin{scope}[xshift=0cm]
    \node (000) at (0,0) {$000$};
    \node (001) at (1,1) {$001$};
    \node (010) at (0,2) {$010$};
    \node (011) at (1,3) {$011$};
    \node (100) at (2,0) {$100$};
    \node (101) at (3,1) {$101$};
    \node (110) at (2,2) {$110$};
    \node (111) at (3,3) {$111$};

    \path[draw] (000) -- (001) -- (011) -- (111)
    (000) -- (010) -- (110) -- (111)
    (000) -- (100) -- (101) -- (111)
    (001) -- (101)
    (010) -- (011)
    (100) -- (110);
    
    \draw[very thick, blue, -latex] (010) -- (110);
    \draw[very thick, blue, -latex] (000) -- (100);
    \draw[very thick, blue, -latex] (111) -- (110);
    \draw[very thick, blue, -latex] (101) -- (100);

    \draw[very thick, blue, latex-latex] (010) -- (000);
    \draw[very thick, blue, latex-latex] (110) -- (100);
    \draw[very thick, blue, latex-latex] (111) -- (101);

    \draw[very thick, blue, -latex] (001) -- (011);
\end{scope}

\begin{scope}[xshift=5cm]
    \node (000) at (0,0) {$000$};
    \node (001) at (1,1) {$001$};
    \node (010) at (0,2) {$010$};
    \node (011) at (1,3) {$011$};
    \node (100) at (2,0) {$100$};
    \node (101) at (3,1) {$101$};
    \node (110) at (2,2) {$110$};
    \node (111) at (3,3) {$111$};

    \path[draw] (000) -- (001) -- (011) -- (111)
    (000) -- (010) -- (110) -- (111)
    (000) -- (100) -- (101) -- (111)
    (001) -- (101)
    (010) -- (011)
    (100) -- (110);
    
    \draw[very thick, blue, latex-latex] (000) -- (100);
    \draw[very thick, blue, latex-latex] (010) -- (110);
    
    \draw[very thick, blue, -latex] (000) -- (010);
    \draw[very thick, blue, -latex] (100) -- (110);

    \draw[very thick, blue, -latex] (011) -- (111);
    \draw[very thick, blue, -latex] (101) -- (111);
\end{scope}

\begin{scope}[xshift=10cm]
    \node (000) at (0,0) {$000$};
    \node (001) at (1,1) {$001$};
    \node (010) at (0,2) {$010$};
    \node (011) at (1,3) {$011$};
    \node (100) at (2,0) {$100$};
    \node (101) at (3,1) {$101$};
    \node (110) at (2,2) {$110$};
    \node (111) at (3,3) {$111$};

    \path[draw] (000) -- (001) -- (011) -- (111)
    (000) -- (010) -- (110) -- (111)
    (000) -- (100) -- (101) -- (111)
    (001) -- (101)
    (010) -- (011)
    (100) -- (110);
    
    \draw[very thick, blue, -latex] (010) -- (110);
    \draw[very thick, blue, -latex] (000) -- (100);
    \draw[very thick, blue, -latex] (010) -- (000);
    \draw[very thick, blue, -latex] (110) -- (100);
    \draw[very thick, blue, -latex] (101) -- (100);

    \draw[very thick, blue, latex-latex] (001) -- (011);

\end{scope}

\end{tikzpicture}
\end{center}

We can now give the classification theorem.

\begin{theorem}[Classification of asynchronous graphs of commutative networks \cite{BGT20}]
  \label{theo:boolean_cs}
  A Boolean network is commutative if and only if it is a union of arrangement networks.
\end{theorem}

We shall focus on two special cases of commutative networks. 
On the one hand, a \Define{negation on subcubes} is a union of arrangement networks, where each arrangement $X$ is a single subcube, and $f_i(x) = \neg x_i$ for each $x \in X$. Two examples are displayed below.

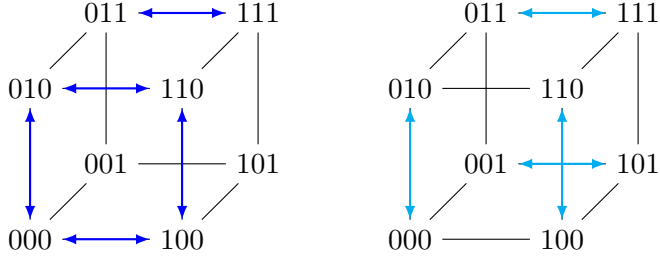
\begin{center}
    
\begin{tikzpicture}
\begin{scope}[xshift=0cm]
    \node (000) at (0,0) {$000$};
    \node (001) at (1,1) {$001$};
    \node (010) at (0,2) {$010$};
    \node (011) at (1,3) {$011$};
    \node (100) at (2,0) {$100$};
    \node (101) at (3,1) {$101$};
    \node (110) at (2,2) {$110$};
    \node (111) at (3,3) {$111$};

    \path[draw] (000) -- (001) -- (011) -- (111)
    (000) -- (010) -- (110) -- (111)
    (000) -- (100) -- (101) -- (111)
    (001) -- (101)
    (010) -- (011)
    (100) -- (110);
    
    \draw[very thick, latex-latex, blue] (000) -- (010);
    \draw[very thick, latex-latex, blue] (010) -- (110);
    \draw[very thick, latex-latex, blue] (110) -- (100);
    \draw[very thick, latex-latex, blue] (100) -- (000);

    \path[draw, thick, latex-latex, blue] (011) -- (111);

\end{scope}

\begin{scope}[xshift=5cm]
    
    \node (000) at (0,0) {$000$};
    \node (001) at (1,1) {$001$};
    \node (010) at (0,2) {$010$};
    \node (011) at (1,3) {$011$};
    \node (100) at (2,0) {$100$};
    \node (101) at (3,1) {$101$};
    \node (110) at (2,2) {$110$};
    \node (111) at (3,3) {$111$};

    \path[draw] (000) -- (001) -- (011) -- (111)
    (000) -- (010) -- (110) -- (111)
    (000) -- (100) -- (101) -- (111)
    (001) -- (101)
    (010) -- (011)
    (100) -- (110);
    
    \path[draw, very thick, latex-latex, cyan] (000) -- (010);
    \path[draw, very thick, latex-latex, cyan] (100) -- (110);
    \path[draw, very thick, latex-latex, cyan] (001) -- (101);
    \path[draw, very thick, latex-latex, cyan] (011) -- (111);
\end{scope}

\end{tikzpicture}
\end{center}

On the other hand, a \Define{constant on arrangements} is a union of arrangements, where for each arrangement $X$, $f(x) = f(y)$ for all $x,y \in X$.
Two examples are displayed below.

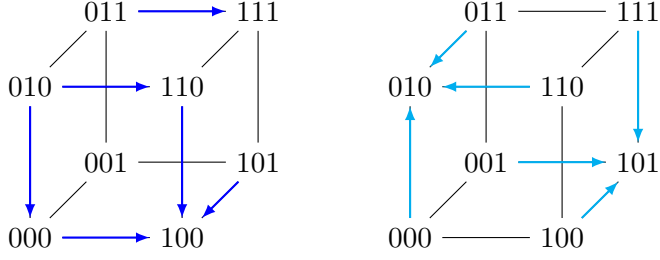
\begin{center}
    
\begin{tikzpicture}
\begin{scope}[xshift=0cm]
    \node (000) at (0,0) {$000$};
    \node (001) at (1,1) {$001$};
    \node (010) at (0,2) {$010$};
    \node (011) at (1,3) {$011$};
    \node (100) at (2,0) {$100$};
    \node (101) at (3,1) {$101$};
    \node (110) at (2,2) {$110$};
    \node (111) at (3,3) {$111$};

    \path[draw] (000) -- (001) -- (011) -- (111)
    (000) -- (010) -- (110) -- (111)
    (000) -- (100) -- (101) -- (111)
    (001) -- (101)
    (010) -- (011)
    (100) -- (110);
    
    \draw[very thick, -latex, blue] (010) -- (110);
    \draw[very thick, -latex, blue] (010) -- (000);
    \draw[very thick, -latex, blue] (000) -- (100);
    \draw[very thick, -latex, blue] (110) -- (100);
    \draw[very thick, -latex, blue] (101) -- (100);

    \draw[very thick, -latex, blue] (011) -- (111);

\end{scope}

\begin{scope}[xshift=5cm]
    \node (000) at (0,0) {$000$};
    \node (001) at (1,1) {$001$};
    \node (010) at (0,2) {$010$};
    \node (011) at (1,3) {$011$};
    \node (100) at (2,0) {$100$};
    \node (101) at (3,1) {$101$};
    \node (110) at (2,2) {$110$};
    \node (111) at (3,3) {$111$};

    \path[draw] (000) -- (001) -- (011) -- (111)
    (000) -- (010) -- (110) -- (111)
    (000) -- (100) -- (101) -- (111)
    (001) -- (101)
    (010) -- (011)
    (100) -- (110);
    
    \path[draw, latex-, very thick, cyan] (010) -- (110);
    \path[draw, latex-, very thick, cyan] (010) -- (000);
    \path[draw, latex-, very thick, cyan] (010) -- (011);
    
    \path[draw, latex-, very thick, cyan] (101) -- (001);
    \path[draw, latex-, very thick, cyan] (101) -- (111);
    \path[draw, latex-, very thick, cyan] (101) -- (100);
\end{scope}
\end{tikzpicture}
\end{center}

We immediately see that negations on subcube are exactly the commutative networks with symmetric asynchronous graphs, and similarly the constants on arrangements are exactly the commutative networks with oriented asynchronous graphs.

\subsection{Commutative networks are trapping} \label{subsection:commutative_trapping}

This subsection is devoted to comparing commutative networks to trapping networks. Recall (Theorem \ref{theorem:trapping_networks}) that a network is trapping if and only if for any $x \in \B^n$ and any $y \in [x, f(x)]$,
\[
    [y, f(y)] \subseteq [x, f(x)].
\]

\begin{theorem}[Alternate definitions of commutative networks] \label{theorem:commutative_trapping}
Any commutative network is trapping. More precisely, the following are equivalent for $f \in \Functions( n )$:
\begin{enumerate}
    \item \label{item:commutative_definition}
    $f$ is commutative, i.e. $f^{ (i,j) } = f^{ (j,i) }$ for all $i,j \in [n]$.

    \item \label{item:commutative_xy}
    $[y, f(x)] \subseteq [y, f(y)] \subseteq [x, f(x)]$ for all $x \in \B^n$ and any $y \in [x, f(x)]$.

    \item \label{item:commutative_ST}
    $f^{ (S \Delta T) } \sqsubseteq f^{( S,T )} \sqsubseteq f^{( S \cup T )}$ for all $S, T \subseteq [n]$.
\end{enumerate}
\end{theorem}

\begin{proof}
$\ref{item:commutative_definition} \implies \ref{item:commutative_xy}$. Let $f$ be a commutative network, $x \in \B^n$ and $y \in [x, f(x)]$. Denote $y = f^{(S)}(x)$ for some $S = \Delta( x, y ) \subseteq \Delta( x, f(x) )$. We first prove $[y, f(y)] \subseteq [x, f(x)]$. For all $i \notin \Delta( x, f(x) )$, we have
\[
    f_i( y ) = f_i( f^{(S)}( x ) ) = f_i( x ) = x_i = y_i,
\]
hence $i \notin \Delta( y, f(y) )$. We now prove that $[y, f(x)] \subseteq [y, f(y)]$. For all $j \in \Delta( y, f(x) ) = \Delta( x, f(x) ) \setminus S$, we have
\[
    f_j( y ) = f_j( f^{(S)}( x ) ) = f_j( x ) \ne x_j = y_j,
\]
hence $j \in \Delta( y, f(y) )$.

$\ref{item:commutative_xy} \implies \ref{item:commutative_definition}$. Suppose, for the sake of contradiction, that $f$ is not commutative, i.e. there exist $i, j \in [n]$ and $x \in \B^n$ such that $f_i( f^{(j)}( x ) ) \ne f_i( x )$. Denoting $y = f^{(j)}( x )$, we have $y \in [x, f(x)]$. If $i \in \Delta( x, f(x) )$, then $\Delta( y, f(x) ) \not\subseteq \Delta( y, f(y) )$; if $i \notin \Delta( x, f(x) )$, then $\Delta( y, f(y) ) \not\subseteq \Delta( x, f(x) )$. In either case, we obtain a contradiction.

$\ref{item:commutative_definition} \implies \ref{item:commutative_ST}$. Since $f$ is trapping, we have $f^{( S,T )} \sqsubseteq f^{(S \cup T)}$. Moreover, by commutativity,
\[
    f^{(S,T)} = f^{( S \Delta T, S \cap T, S \cap T )},
\]
and hence $f^{(S,T)}_{S \Delta T} = f_{S \Delta T}$, which implies $f^{ (S \Delta T) } \sqsubseteq f^{( S,T )}$.

$\ref{item:commutative_ST} \implies \ref{item:commutative_definition}$. If $S \cap T = \emptyset$, we have $S \Delta T = S \cup T$ and hence $f^{( S,T )} = f^{ (S \cup T) }$. Therefore, $f$ is commutative.
\end{proof}

\subsection{Collections of principal trapspaces of commutative networks}

We now classify the collections of principal trapspaces of commutative networks. Say a collection of subcubes $\mathcal{A} \in \Collections( n )$ is \Define{convex} if the following holds: for all $Q, R \in \mathcal{A}$ and any $S$ with $Q \subseteq S \subseteq R$, $S \in \mathcal{A}$.

\begin{theorem}[Commutative networks and convex collections of principal trapspaces] \label{theorem:convex_principal_trapspaces}
A pre-principal collection of subcubes is the collection of principal trapspaces of a commutative network if and only if it is convex.
\end{theorem}

We begin the proof with a technical lemma. For any $f \in \Functions(n)$ and any $x \in \B^n$, denote the dimension of $T_f(x)$ as $\delta_x = \HammingDistance( x, f(x) )$.

\begin{lemma} \label{lemma:delta}
Let $f$ be a commutative network.  For any $y \in [x, f(x)]$, we have 
\[
    \HammingDistance(x, y) \ge \delta_x - \delta_y \ge 0,    
\]
and $\HammingDistance(x, y) = \delta_x - \delta_y$ if and only if $f(y) = f(x)$.
\end{lemma}

\begin{proof}
Let $x \in \B^n$ and $y \in [x, f(x)]$. Firstly, since $f(y) \in [x, f(x)]$, we immediately obtain $\delta_x \ge \delta_y$. Secondly, we have
\[
    \delta_y + \HammingDistance( x, y ) = \HammingDistance( y, f(y) ) + \HammingDistance( x, y ) \ge \HammingDistance( y, f(x) ) + \HammingDistance( x, y ) \ge \HammingDistance( x, f(x) ) = \delta_x,
\]
where we used Theorem \ref{theorem:commutative_trapping} to show $\HammingDistance( y, f(y) ) \ge \HammingDistance( y, f(x) )$. Thirdly, by the above, we have $\HammingDistance(x, y) = \delta_x - \delta_y$ only if $\HammingDistance( y, f(y) ) = \HammingDistance( y, f(x) )$. Since $f(x) \in [y, f(y)]$, this is equivalent to $f(x) = f(y)$. Finally, if $f(x) = f(y)$, then 
\[
    \delta_x = \HammingDistance( x, f(x) ) = \HammingDistance( x, y ) + \HammingDistance( y, f(x) ) = \HammingDistance( x, y ) + \HammingDistance( y, f(y) ) = \HammingDistance( x, y ) + \delta_y.
\]
\end{proof}

\begin{proof}[Proof of Theorem \ref{theorem:convex_principal_trapspaces}]
Let $f$ be a commutative network. Now, suppose that $T_f( y ) \subseteq T_f( x )$ and that $x$ and $y$ are nearest possible, i.e. if $T_f( x' ) = T_f( x )$ and $T_f( y' ) \subseteq T_f( y )$, then $\HammingDistance( x, y ) \le \HammingDistance( x', y' )$.

\begin{claim}
With the conditions above, $f(x) = f(y)$.
\end{claim}

\begin{proof}
We first prove that $y \in [x, f(y)]$. If there exists $i \in \Delta( y, f(y) ) \cap \Delta( x, y )$, then $y' = y + e_i$ satisfies $T_f( y' ) \subseteq T_f( y )$ and $\HammingDistance( x, y' ) = \HammingDistance( x, y ) - 1$, which contradicts our assumptions. Therefore, $\Delta( y, f(y) ) \cap \Delta( x, y ) = \emptyset$. We obtain that for all $j \in [n]$, 
\[
    y_j \ne x_j \implies y_j = f( y )_j \ne x_j,
\]
which is equivalent to $\Delta( y, x ) \subseteq \Delta( x, f(y) )$. Similarly, we obtain $\Delta( y, f(y) ) \subseteq \Delta( x, f(y) )$. Thus, $y \in [x, f(y)]$.

Now, $f(x) \in [y, f(y)]$ because $f$ is commutative, hence
\[
    [ x, f(y) ] = [ x, y, f(y) ] = [ x, y, f(y), f(x) ] = [x, f(x)],
\]
thus $f(y) = f(x)$.
\end{proof}

Any subcube $S$ satisfying $T_f(y) = [y, f(x)] \subseteq S \subseteq [x, f(x)] = T_f(x)$ can be expressed as $S = [z, f(x)]$ for some $z \in [x, y]$. Therefore, we only need to prove that $f(z) = f(x)$ for all $z \in [x, y]$. Since $z = f^{(\Delta(x, z))}( x )$, we obtain
\[
    f^{( \Delta(z,y) )}( z ) = f^{( \Delta(z,y) )}( f^{( \Delta(x, z) )} ( x ) ) = f^{(\Delta(x,y))}( x ) = y.
\]
Therefore, $y \in [z, f(z)]$. We now repeatedly use Lemma \ref{lemma:delta}. Since $f(y) = f(x)$, we have $\HammingDistance( x, y ) = \delta_x - \delta_y$. We obtain
\[
    \HammingDistance( x, y ) = \delta_x - \delta_y = ( \delta_x - \delta_z ) + ( \delta_z - \delta_y ) \le \HammingDistance( x, z ) + \HammingDistance( z,y ) = \HammingDistance( x, y ),
\]
thus $\HammingDistance( x, z ) = \delta_x - \delta_z$, and hence $f(z) = f(x)$, and we are done.

\medskip

Conversely, suppose $\mathcal{A}$ is a convex pre-principal collection of subcubes and let $f = F_\PrincipalTrapspaces( \mathcal{A} )$; recall that $f$ is trapping. We only need to prove that $f(x) \in [y, f(y)]$ for all $x \in \B^n$, $y \in [x, f(x)]$. Suppose, for the sake of contradiction, that $f(x) \notin [y, f(y)]$, say $f_j(x) \ne y_j = f_j(y)$ for some $j \in \Delta( x, f(x) )$. Then
\[
    [y, f(y)] \subseteq [x, y, f(y)] \subseteq [x, f(x)],
\]
and by convexity, $[x, y, f(y)]$ is a principal trapspace of $f$, that must contain $f(x)$. But $f_j(x) \notin \{ x_j, y_j, f_j(y) \}$, hence $f(x) \notin [x, y, f(y)]$, which is the desired contradiction.
\end{proof}

\section{Marseille networks} \label{section:Marseille_networks}

\subsection{Alternate definitions} \label{subsection:Marseille_definitions}

A network is \Define{Marseille} if it is commutative and bijective. As seen above, this is equivalent to commutative and locally bijective, and equivalent to commutative and globally bijective.
As expected, Marseille networks are exactly the negations on subcubes.

\begin{theorem}[Alternate definitions of Marseille networks] \label{theorem:Marseille_definitions}
The following are equivalent for $f \in \Functions( n )$:
\begin{enumerate}
    \item \label{item:Marseille_definition}
    $f$ is Marseille, i.e. $f$ is commutative and bijective.

    \item \label{item:Marseille_negation_on_subcubes}
    $f$ is a negation on subcubes.

    \item \label{item:Marseille_xy}
    for all $x \in \B^n$ and $y \in [x, f(x)]$, $[y, f(y)] = [x, f(x)]$.

    \item \label{item:Marseille_ST}
    for all $S, T \subseteq [n]$, $f^{ (S \Delta T) } = f^{( S,T )}$.
\end{enumerate}
\end{theorem}

\begin{proof}
$\ref{item:Marseille_definition} \implies \ref{item:Marseille_xy}$. Since $f$ is commutative, we have $[ y, f(y) ] \subseteq [ x, f(x) ]$. Suppose $[ y, f(y) ] \subset [ x, f(x) ]$ for some $y \in [x, f(x)]$; then $x \notin [y, f(y)]$. We have $y = f^{(S)}( x )$ for some $S \subseteq [n]$; therefore $f^{(S)}( T_f(y) \cup \{ x \} ) \subseteq T_f(y)$, which contradicts the fact that $f^{(S)}$ is bijective.

$\ref{item:Marseille_definition} \iff \ref{item:Marseille_negation_on_subcubes}$. This easily follows from Theorem \ref{theo:boolean_cs}.

$\ref{item:Marseille_xy} \implies \ref{item:Marseille_definition}$. By Theorem \ref{theorem:commutative_trapping}, $f$ is commutative. 

For the sake of contradiction, suppose that $f$ is not bijective, i.e. there exist distinct $x, y \in \B^n$ such that $f( x ) = f( y )$. If $y \in [x, f(x)]$, we have $[y, f(x)] = [y, f(y)] = [x, f(x)]$ and hence $x = y$. Therefore, $y \notin T_f(x)$ and $x \notin T_f(y)$. But then, $z = f(x)$ satisfies $z \in [x, f(x)]$ and $f(z) \in T_f(z) \subseteq T_f( x ) \cap T_f( y ) \subset T_f(x)$, hence $[z, f(z)] \ne [x, f(x)]$, which is the desired contradiction.

$\ref{item:Marseille_definition} \implies \ref{item:Marseille_ST}$. We have
\[
    f^{(S, T)} = f^{( S \Delta T, S \cap T, S \cap T )} = f^{(S \Delta T)}.
\]

$\ref{item:Marseille_ST} \implies \ref{item:Marseille_definition}$. $f$ is clearly commutative, and $f^{(S, S)} = \id$ for all $S \subseteq [n]$, i.e. $f$ is globally bijective. 
\end{proof}

\subsection{Relations amongst the three graphs: the symmetric case}

Recall that a graph $\Gamma$ is symmetric if $(u,v) \in E$ implies $(v,u) \in E$. We now investigate networks with symmetric (asynchronous, general asynchronous, trapping) graphs.

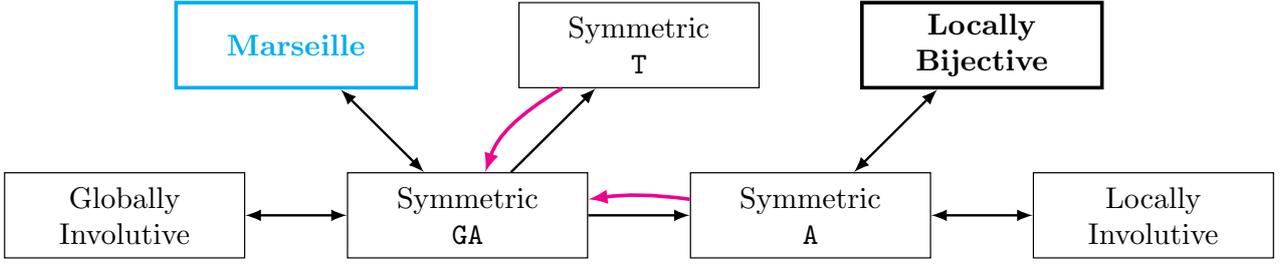
\begin{figure}
\centering
\resizebox{\textwidth}{!}{
\begin{tikzpicture}[yscale=2, xscale=4, font=\small]

\begin{scope}[yshift=0cm]

    \node[draw, minimum height=1cm, minimum width=2.8cm, align=center] (GI) at (-0.5,0) {Globally\\ Involutive};
    \node[draw, minimum height=1cm, minimum width=2.8cm, align=center] (SGA) at (0.5,0) {Symmetric\\ $\GeneralAsynchronous$};
    \node[draw, very thick, minimum height=1cm, minimum width=2.8cm, align=center, cyan] (M) at (0,1) {\textbf{Marseille}};
    
    \node[draw, minimum height=1cm, minimum width=2.8cm, align=center] (ST) at (1,1) {Symmetric\\ $\TrappingGraph$};

    \node[draw, minimum height=1cm, minimum width=2.8cm, align=center] (SA) at (1.5,0) {Symmetric \\ $\Asynchronous$};
    \node[draw, minimum height=1cm, minimum width=2.8cm, align=center] (LI) at (2.5,0) {Locally\\ Involutive};
    \node[draw, very thick, minimum height=1cm, minimum width=2.8cm, align=center] (LB) at (2,1) {\textbf{Locally}\\ \textbf{Bijective}};

    \draw[thick, latex-latex] (SGA) -- (M);
    \draw[thick, latex-latex] (SGA) -- (GI);
    
    \draw[thick, -latex] (SGA) -- (ST);
    \draw[-latex, very thick, magenta] (ST) to [bend right=15] (SGA);
    
    \draw[thick, -latex] (SGA) -- (SA);
    \draw[-latex, very thick, magenta] (SA) to [bend right=15] (SGA);
    
    \draw[thick, latex-latex] (SA) -- (LB);
    \draw[thick, latex-latex] (SA) -- (LI);
    
\end{scope}
\end{tikzpicture}
}
\caption{Relationships amongst the three main graphs: the symmetric case.} \label{figure:graphs_symmetric}
\end{figure}

We shall represent the implications amongst different properties related to symmetric graphs in the diagram in Figure \ref{figure:graphs_symmetric}. We shall make use of such diagrams in Figures \ref{figure:Marseille}, \ref{figure:graphs_triangular} and \ref{figure:Lille} as well. In such diagrams, black arrows represent implications that hold for all Boolean networks, \textcolor{magenta}{magenta} arrows represent implications that hold for all trapping networks, and \textcolor{blue}{blue} arrows represent implications that hold for all commutative networks. 

We say that such a diagram is \Define{correct} if all the implications depicted indeed hold; we say that it is \Define{complete} if any other implication (that might hold for all networks, or all trapping networks, or all commutative networks) between the different properties of the diagram can be inferred from the diagram. For instance, in Figure \ref{figure:Marseille}, the implication ``Commutative and Bijective $\implies$ Marseille'' can be inferred from the graph by following two arrows; on the other hand, there exists a counterexample to the implication ``Trapping and Bijective $\implies$ Locally Bijective.''

When proving the correctness of a diagram, we prove a sufficient subset of implications, such that any other implication follows from one in that subset. Similarly, when proving the completeness of a diagram, we exhibit a counterexample for each implication in a sufficient subset of incorrect implications.

The diagram in Figure \ref{figure:graphs_symmetric} not only gives relations amongst symmetric asynchronous, general asynchronous, and trapping graphs. It also provides several alternate definitions of Marseille networks.

\begin{theorem}[Graphs defined from networks: Symmetric graphs] \label{theorem:graphs_symmetric}
The diagram in Figure \ref{figure:graphs_symmetric} is correct and complete.
\end{theorem}

\begin{proof}
Firstly, we prove the triple equivalence Marseille $\iff$ Globally Involutive $\iff$ Symmetric $\GeneralAsynchronous$.

\begin{enumerate}
    \item Marseille $\implies$ Globally Involutive. Follows from the results on commutative networks reviewed in Section \ref{subsection:review_commutative}.

    \item Globally Involutive $\implies$ Symmetric $\GeneralAsynchronous$. If $x \to_{\GeneralAsynchronous(f)} x$ but $y \not\to_{\GeneralAsynchronous(f)} x$, we have $y = f^{(S)}( x )$ for $S = \Delta( x, y )$ but $f^{(S)}( y ) \ne x$, hence $f^{(S)}$ is not involutive.

    \item Symmetric $\GeneralAsynchronous$ $\implies$ Marseille. Let $y \in [x, f(x)]$ and $\Gamma = \GeneralAsynchronous( f )$. Let us prove that $[x, f(x)] \subseteq [y, f(y)]$. We have $x \to_\Gamma y$, hence by symmetry $y \to_\Gamma x$ (or in other words, $x \in [y, f(y)]$). Moreover, we have $x \to_\Gamma f(x)$, hence $f(x) \to_\Gamma x$ and $f(x) \to_\Gamma y$; and by symmetry $y \to_\Gamma f(x)$ (or in other words, $f(x) \in [y, f(y)]$). We obtain $[x, f(x)] \subseteq [y, f(y)]$, as desired. Now, since $x \in [y, f(y)]$, we obtain $[y, f(y)] \subseteq [x, f(x)]$; therefore $[x, f(x)] = [y, f(y)]$ and $f$ is Marseille.
\end{enumerate}

Secondly, we prove the other black implications and equivalences in Figure \ref{figure:graphs_symmetric}.
\begin{enumerate}[resume]
    \item Symmetric $\GeneralAsynchronous$ $\implies$ Symmetric $\TrappingGraph$. $f$ is Marseille, hence it is trapping and $\TrappingGraph(f) = \GeneralAsynchronous(f)$.

    \item Symmetric $\GeneralAsynchronous$ $\implies$ Symmetric $\Asynchronous$. For any distinct $x, y \in \B^n$, we have 
    \[
        x \to_{\Asynchronous(f)} y \implies x \to_{\GeneralAsynchronous(f)} y \text{ and } \HammingDistance(x,y) = 1 \implies y \to_{\GeneralAsynchronous(f)} x \text{ and } \HammingDistance(x,y) = 1 \implies y \to_{\Asynchronous(f)} x.
    \]

    \item Locally Bijective $\iff$ Locally Involutive. $f^{(i)}$ is dynamically local, hence it is bijective if and only if it is involutive.

    \item Locally Bijective $\iff$ Symmetric $\Asynchronous$. Let $f \in \Functions(n)$ and $i \in [n]$. The function $f^{ (i) }$ can be decomposed into $2^{n-1}$ functions, one for each value of $x_{-i}$. More formally, for any $z \in \B^{n-1}$, let $g^z : \B \to \B$ be defined by $g^z( a ) = f_i(a, z)$ for all $a \in \B$. For all $z \in \B^{n-1}$, let $\Gamma^z$ be the subgraph of $\Asynchronous(f)$ induced by $\{ x = (x_i = 0, x_{-i} = z), y = (y_i = 1, y_{-i} = z) \}$. We note that $g^z$ is bijective if and only if $\Gamma^z$ is symmetric (either $g^z$ is the identity, in which case $\Gamma^z$ has two loops, or $g^z$ is the transposition $(0, 1)$, in which case $\Gamma^z$ is complete). We then have $f^{ (i) }(x) = ( g^{ x_{-i} }( x_i ), x_{-i} )$, so that $f^{ (i) }$ is bijective if and only if all $g^z$ functions are bijective. Thus, $f$ is locally bijective if and only if $\Asynchronous(f)$ is symmetric. 
\end{enumerate}

Thirdly, we prove the magenta implications in Figure \ref{figure:graphs_symmetric}, that only hold for trapping networks.
\begin{enumerate}[resume]
    \item Trapping and Symmetric $\TrappingGraph$ $\implies$ Symmetric $\GeneralAsynchronous$. We have $\TrappingGraph(f) = \GeneralAsynchronous(f)$.

    \item Trapping and Symmetric $\Asynchronous$ $\implies$ Marseille. We prove that for all $y \in [x, f(x)]$, $[y, f(y)] = [x, f(x)]$. The proof is by induction on the distance $\HammingDistance( x, y )$, and is clear for distance $0$. Suppose it holds for distance $d$, and suppose $\HammingDistance( x, y ) = d + 1$. Let $z \in [x, y] \subseteq [x, f(x)]$ such that $\HammingDistance( x, z ) = d$ and $\HammingDistance( z, y ) = 1$. By induction hypothesis, $[z, f(z)] = [x, f(x)]$ and hence $z \to_{\Asynchronous(f)} y$; by symmetry, $y \to_{\Asynchronous(f)} z$, hence $[x, f(x)] = [z, f(z)] = T_f(z) \subseteq T_f(y) = [y, f(y)]$.

\end{enumerate}

Finally, we exhibit counterexamples to implications that are not displayed in Figure \ref{figure:graphs_symmetric}.
\begin{enumerate}[(a)]
    \item \label{item:counterexample_a} Symmetric $\TrappingGraph$ $\centernot\implies$ Symmetric $\Asynchronous$.

    \item \label{item:counterexample_b} Symmetric $\Asynchronous$ $\centernot\implies$ Symmetric $\TrappingGraph$.
\end{enumerate}

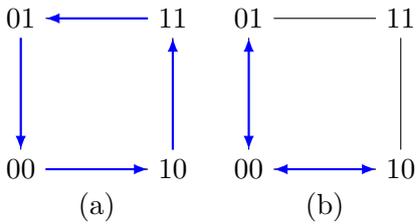
\begin{center}
    
\begin{tikzpicture}
    \begin{scope}[xshift=0cm, yshift=0cm]
    \node (a) at (1,-0.5) {\ref{item:counterexample_a}};
    \node (00) at (0,0) {$00$};
    \node (01) at (0,2) {$01$};
    \node (10) at (2,0) {$10$};
    \node (11) at (2,2) {$11$};

    \path[draw] (00) -- (01) -- (11)
    (00) -- (10) -- (11);
    
    
    \draw[very thick,-latex, blue] (00) -- (10);
    \draw[very thick,-latex, blue] (10) -- (11);  
    \draw[very thick,-latex, blue] (11) -- (01);
    \draw[very thick,-latex, blue] (01) -- (00);
    \end{scope}

    \begin{scope}[xshift=5cm, yshift=0cm]
    \node (b) at (1,-0.5) {\ref{item:counterexample_b}};
    \node (00) at (0,0) {$00$};
    \node (01) at (0,2) {$01$};
    \node (10) at (2,0) {$10$};
    \node (11) at (2,2) {$11$};

    \path[draw] (00) -- (01) -- (11)
    (00) -- (10) -- (11);
    
    
    \draw[very thick,latex-latex, blue] (00) -- (10);
    \draw[very thick,latex-latex, blue] (00) -- (01);
    \end{scope}
    
\end{tikzpicture}   
\end{center}

\end{proof}

\subsection{Classification of Marseille networks}



\begin{figure}
\centering
\resizebox{\textwidth}{!}{
\begin{tikzpicture}[yscale=2, xscale=4, font=\small]

\begin{scope}[yshift=0cm]

    \node[draw, very thick, cyan, minimum height=1cm, minimum width=2.8cm, align=center] (M) at (0,1)  {Marseille}; 

    \node[draw, minimum height=1cm, minimum width=2.8cm, align=center] (I) at (0.5,0) {Involutive};
    
    \node[draw, minimum height=1cm, minimum width=2.8cm, align=center] (B3) at (1,1) {Globally\\ Bijective};
    \node[draw, minimum height=1cm, minimum width=2.8cm, align=center] (B1) at (2,1) {Locally\\ Bijective};
    \node[draw, minimum height=1cm, minimum width=2.8cm, align=center] (B) at (1.5,0) {Bijective};

    \draw[thick, -latex] (M) -- (B3);
    \draw[thick, -latex] (B3) -- (B1);
    \draw[thick, -latex] (B3) -- (B);
    \draw[thick, -latex] (M) -- (I);
    \draw[thick, -latex] (I) -- (B);

    \draw[-latex, very thick, magenta] (B3) to [bend right=15] (M);
    \draw[-latex, very thick, magenta] (B1) to [bend right=15] (B3);
    \draw[-latex, very thick, magenta] (B) to [bend right=15] (I);

    \draw[-latex, thick, blue] (B) to [bend right=15] (B3);
    \draw[-latex, thick, blue] (I) to [bend right=15] (M);

\end{scope}
    
\end{tikzpicture}
}
\caption{\textcolor{cyan}{Marseille} v \textcolor{blue}{commutative} and \textcolor{magenta}{trapping} networks.} \label{figure:Marseille}
\end{figure}
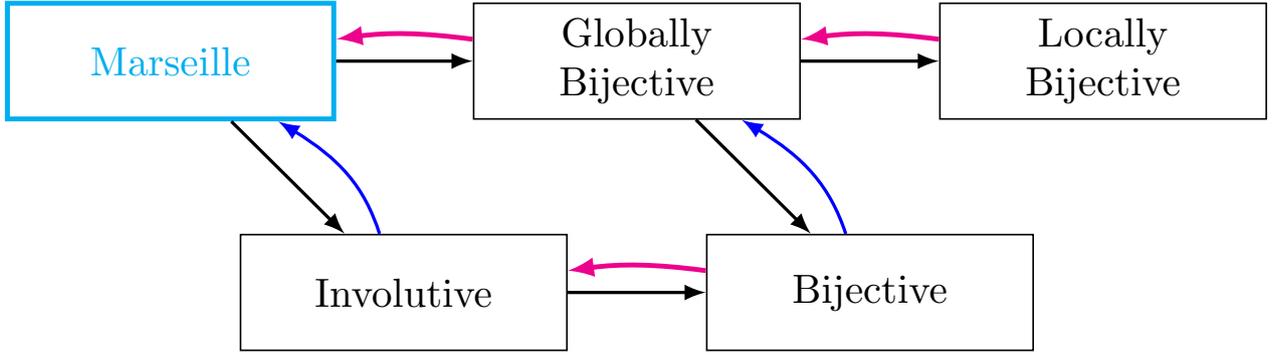

The forthcoming classification of Marseille networks as trapping networks will involve proving that bijective trapping networks are involutive. For the sake of completeness, we prove that all trapping networks have a period of at most $2$, but their transient length can be up to $n$.

\begin{proposition} \label{proposition:transient_length_period_trapping}
Let $f \in \Trapping{\Functions}(n)$, then $f^{n+2} = f^n$. Moreover, for all $n \ge 3$, there exists a trapping network with period $2$ and transient length $n$.
\end{proposition}

\begin{proof}
For any $x$, let $O(x) = \{ f^i(x) : i \in \mathbb{N} \}$ be the orbit of $x$. The sequence $T_i := T_f( f^i(x) )$ for $i \in \mathbb{N}$ is a descending chain of subcubes. If $T_i = T_{i+1}$, then we have $[f^i(x), f^{i+1}(x)] = [f^{i+1}(x), f^{i+2}(x)]$, hence $f^i(x) = f^{i+2}(x)$ and $T_i = T_{i+1} = \dots = T_n$, and $|O(x)| \le i+1$. Since $T_0$ has dimension $n$, we obtain $T_n = T_{n+1}$ and hence $|O(x)| \le n+1$. Therefore, $f$ has transient length at most $n$. Moreover, if $x$ is a periodic point, say $f^k(x) = x$ we have $T_0 = T_{k+1}$, hence $T_0 = T_1$ and $x = f^2(x)$.

Conversely, the trapping network $f$ with period $2$ and transient length $n \ge 3$ is constructed as follows. First, let $t^1, \dots, t^{n+1} \in \B^n$ be defined as
\[
    t^i_j = \begin{cases}
    1 & \text{if } j < i \\
    i + j \mod 2 & \text{otherwise}.
    \end{cases}
\]
For instance, for $n = 4$ we obtain
\[
    t^1 = 0101, \; t^2 = 1010, \; t^3 = 1101, \; t^4 = 1110, \; t^5 = 1111.
\]
Second, let $c^1 = 0 \dots 00$ and $c^2 = 0 \dots 01$; note that $c^i \ne t^j$ for all $i \in \{1,2\}$ and $j \in \{1, \dots, n+1\}$. Third, let
\[
    f(x) = \begin{cases}
    t^{i+1} &\text{if } x = t^i, i \le n, \\
    c^2 &\text{if } x = c^1, \\
    c^1 &\text{if } x = c^2, \\
    x &\text{otherwise}.
    \end{cases}
\]
Then clearly, $f$ has period $2$ and transient length $n$. It is easily shown that $f$ is also trapping.
\end{proof}

We now classify Marseille networks as specific trapping networks.

\begin{theorem}[Classification of Marseille networks] \label{theorem:Marseille_classification}
The diagram in Figure \ref{figure:Marseille} is correct and complete.
\end{theorem}

\begin{proof}
The implications in the top row of Figure \ref{figure:Marseille} all follow from previous results. We now prove the remaining implications.

\begin{enumerate}

    \item Marseille $\implies$ Involutive. Marseille is equivalent to Globally Involutive.

    \item Commutative and Involutive $\implies$ Marseille. Trivial.

    \item Involutive $\implies$ Bijective. Trivial.

    \item Trapping and Bijective $\implies$ Involutive. Follows from Proposition \ref{proposition:transient_length_period_trapping}.
\end{enumerate}

We now exhibit counterexamples to implications that are not displayed in Figure \ref{figure:Marseille}.
\begin{enumerate}[(a)]

    \item \label{item:counterexample_c} Trapping and Involutive $\centernot\implies$ Locally Bijective.

    \item \label{item:counterexample_d} Bijective $\centernot\implies$ Involutive.
    
\end{enumerate}

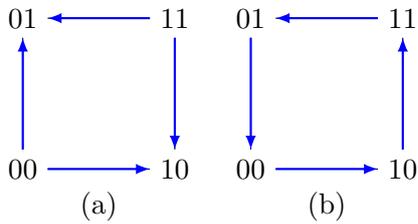
\begin{center}
    
\begin{tikzpicture}

    \begin{scope}[xshift=0cm, yshift=0cm]
    \node (c) at (1,-0.5) {\ref{item:counterexample_c}};
    \node (00) at (0,0) {$00$};
    \node (01) at (0,2) {$01$};
    \node (10) at (2,0) {$10$};
    \node (11) at (2,2) {$11$};

    \path[draw] (00) -- (01) -- (11)
    (00) -- (10) -- (11);
    
    
    \draw[very thick,-latex, blue] (00) -- (01);
    \draw[very thick,-latex, blue] (00) -- (10);  
    \draw[very thick,-latex, blue] (11) -- (10);
    \draw[very thick,-latex, blue] (11) -- (01);
    \end{scope}

    \begin{scope}[xshift=5cm, yshift=0cm]
    \node (d) at (1,-0.5) {\ref{item:counterexample_d}};
    \node (00) at (0,0) {$00$};
    \node (01) at (0,2) {$01$};
    \node (10) at (2,0) {$10$};
    \node (11) at (2,2) {$11$};

    \path[draw] (00) -- (01) -- (11)
    (00) -- (10) -- (11);
    
    
    \draw[very thick,-latex, blue] (00) -- (10);
    \draw[very thick,-latex, blue] (10) -- (11);  
    \draw[very thick,-latex, blue] (11) -- (01);
    \draw[very thick,-latex, blue] (01) -- (00);
    \end{scope}
    
\end{tikzpicture}   
\end{center}

\end{proof}

\section{Lille networks} \label{section:Lille}

\subsection{Alternate definitions} \label{subsection:Lille_definitions}

A network is \Define{Lille} if it is commutative and idempotent. We begin this section with four alternative definitions of Lille networks.

\begin{theorem}[Alternative definitions of Lille networks] \label{theorem:Lille_definitions}
The following are equivalent for $f \in \Functions( n )$.
\begin{enumerate}
    \item \label{item:Lille_definition} 
    $f$ is Lille, i.e. it is commutative and idempotent.

    \item \label{item:Lille_converges}
    $f$ is a constant on arrangements.

    \item \label{item:Lille_xy}
    for all $x \in \B^n$ and $y \in [x, f(x)]$, $[y, f(y)] = [y, f(x)]$.

    \item \label{item:Lille_ST}
    for all $S, T \subseteq [n]$, $f^{( S,T )} = f^{ (S \cup T) }$.
    
\end{enumerate}
\end{theorem}

\begin{proof}
$\ref{item:Lille_definition} \iff \ref{item:Lille_converges}$. Follows from Theorem \ref{theo:boolean_cs}.

$\ref{item:Lille_definition} \implies \ref{item:Lille_xy}$. If $y \in [x, f(x)]$, we have $y = f^{(S)}( x )$ for some $S \subseteq [n]$. By idempotence of $f^{(S)}$, we have
\[
    f( y ) = f^{(S, [n] \setminus S)}( y ) = f^{(S, S, [n] \setminus S)}( x ) = f^{(S, [n] \setminus S)}( x ) = f( x ).
\]

$\ref{item:Lille_xy} \implies \ref{item:Lille_definition}$. $f$ is commutative from Theorem \ref{theorem:commutative_trapping}. Suppose, for the sake of contradiction, that $f^{(S)}$ is not idempotent for some $S \subseteq [n]$, i.e. $f^{(S)}(x)  = y \ne f^{(S)}(y)$. Then $y \in [x, f(x)]$ and yet $f_S( y ) \ne f_S( x )$, and hence $f(y) \ne f(x)$, which is the desired contradiction.

$\ref{item:Lille_definition} \implies \ref{item:Lille_ST}$. We have 
\[
    f^{(S, T)} = f^{( S \Delta T, S \cap T, S \cap T )} = f^{(S \cup T)}.
\]

$\ref{item:Lille_ST} \implies \ref{item:Lille_definition}$. Clearly $f$ is commutative, and $f^{(S, S)} = f^{(S)}$ for all $S \subseteq [n]$, i.e. $f$ is globally idempotent. Thus, $f$ is Lille.
\end{proof}

\subsection{Relation amongst the three main graphs: the triangular case} \label{subsection:graphs_triangular}


A graph is \Define{oriented} if $(u,v) \in E \implies (v,u) \notin E$ for all $u \ne v$. Let us say that a graph is \Define{triangular} if the only cycles in the graph are its loops; in other words, the vertices can be sorted such that the adjacency matrix is triangular.  Say a graph is \Define{sink-terminal} if all its terminal components have cardinality one. 

Let us say a network is \Define{Distinct Principal Trapspaces} (DPT) if $T_f( x ) \ne T_f( y )$ for all $x \ne y$; in other words, $f$ is DPT if and only if $\TrappingGraph( f )$ triangular.
A \Define{fixed point} of $f$ is a configuration $x$ such that $f(x) = x$; we denote the set of fixed points of $f$ by $\Fix( f )$.
We say a network is \Define{Trapspace-FP} if every trapspace of $f$ contains a fixed point. The equivalence between Trapspace-FP and Sink-terminal $\TrappingGraph$ is given in Lemma \ref{lemma:sink_terminal_T} below; its proof is straightforward and hence omitted.
We strengthen this notion in two ways: a network is \Define{Interval-FP} if every interval contains a fixed point; it is \Define{Interval-UFP} if every interval contains a unique fixed point.

Recall that a Boolean network is \Define{fixable} if and only if there is a word $w$ over the alphabet $[n]$ such that $f^w(x)$ is a fixed point for all $x$ \cite{AGRS20}. Equivalently, $f$ is fixable if and only if all its asynchronous attractors are fixed points, i.e. $\Asynchronous( f )$ is sink-terminal.

\begin{lemma} \label{lemma:sink_terminal_T}
Let $f \in \Functions( n )$. The following are equivalent:
\begin{enumerate}
    \item \label{item:sink-terminal_T}
    $\TrappingGraph( f )$ is sink-terminal;

    \item \label{item:smaller_T}
    for any $x \in \B^n \setminus \Fix( f )$, there exists $y \in T_f(x)$ with $T_f( y ) \subset T_f( x )$;

    \item \label{item:mininal_trapspaces}
    the fixed points of $f$ are the only minimal trapspaces of $f$, i.e. $M(f) = \Fix(f)$;

    \item \label{item:fp_in_P}
    any principal trapspace of $f$ contains a fixed point;

    \item \label{item:fp_in_T}
    $f$ is trapspace-FP, i.e. any trapspace of $f$ contains a fixed point.
\end{enumerate}
\end{lemma}

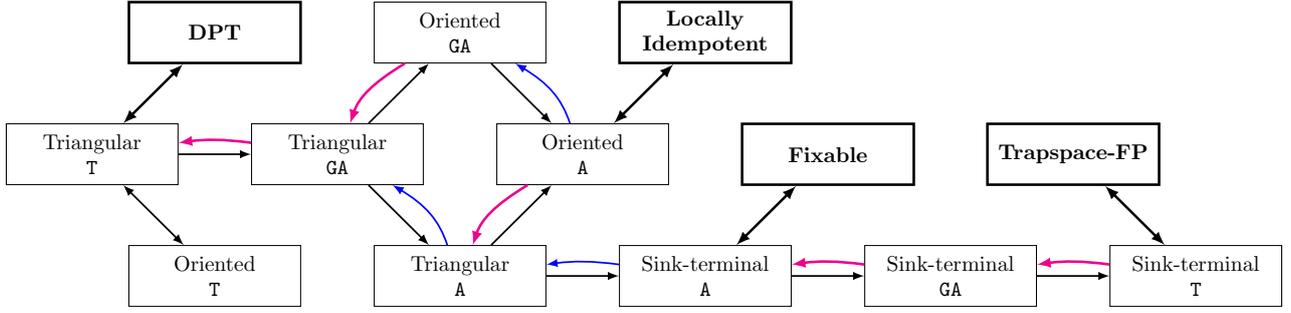
\begin{figure}
\centering
\resizebox{\textwidth}{!}{
\begin{tikzpicture}[yscale=2, xscale=4, font=\small]

    \node[draw, minimum height=1cm, minimum width=2.8cm, align=center] (TT) at (0,1) {Triangular \\ $\TrappingGraph$};
    \node[draw, minimum height=1cm, minimum width=2.8cm, align=center] (OT) at (0.5,0) {Oriented \\ $\TrappingGraph$};
    \node[draw, very thick, minimum height=1cm, minimum width=2.8cm, align=center] (DPT) at (0.5,2) {\textbf{DPT}};

    \node[draw, minimum height=1cm, minimum width=2.8cm, align=center] (TGA) at (1,1) {Triangular \\ $\GeneralAsynchronous$};
    \node[draw, minimum height=1cm, minimum width=2.8cm, align=center] (OGA) at (1.5,2) {Oriented \\ $\GeneralAsynchronous$};

    \node[draw, minimum height=1cm, minimum width=2.8cm, align=center] (TA) at (1.5,0) {Triangular \\ $\Asynchronous$};
    \node[draw, minimum height=1cm, minimum width=2.8cm, align=center] (OA) at (2,1) {Oriented \\ $\Asynchronous$};
    \node[draw, very thick, minimum height=1cm, minimum width=2.8cm, align=center] (LI) at (2.5,2) {\textbf{Locally} \\ \textbf{Idempotent}};
    
    \node[draw, minimum height=1cm, minimum width=2.8cm, align=center] (STA) at (2.5,0) {Sink-terminal \\ $\Asynchronous$};
    \node[draw, very thick, minimum height=1cm, minimum width=2.8cm, align=center] (F) at (3,1) {\textbf{Fixable}};

    \node[draw, minimum height=1cm, minimum width=2.8cm, align=center] (STGA) at (3.5,0) {Sink-terminal \\ $\GeneralAsynchronous$};

    \node[draw, minimum height=1cm, minimum width=2.8cm, align=center] (STT) at (4.5,0) {Sink-terminal \\ $\TrappingGraph$};
    \node[draw, very thick, minimum height=1cm, minimum width=2.8cm, align=center] (TFP) at (4,1) {\textbf{Trapspace-FP}};

    \draw[thick, latex-latex] (TT) -- (OT);

    \draw[thick, -latex] (TGA) -- (OGA);
    \draw[very thick, -latex, magenta] (OGA) to [bend right=15] (TGA);

    \draw[thick, -latex] (TA) -- (OA);
    \draw[very thick, -latex, magenta] (OA) to [bend right=15] (TA);

    \draw[thick, -latex] (TT) -- (TGA);
    \draw[very thick, -latex, magenta] (TGA) to [bend right=15] (TT);
    

    \draw[thick, -latex] (TGA) -- (TA);
    \draw[thick, -latex, blue] (TA) to [bend right=15] (TGA);
    
    \draw[thick, -latex] (OGA) -- (OA);
    \draw[thick, -latex, blue] (OA) to [bend right=15] (OGA);

    \draw[thick, -latex] (TA) -- (STA);
    \draw[thick, -latex, blue] (STA) to [bend right=15] (TA);

    \draw[thick, -latex] (STA) -- (STGA);
    \draw[very thick, -latex, magenta] (STGA) to [bend right=15] (STA);

    \draw[thick, -latex] (STGA) -- (STT);
    \draw[very thick, -latex, magenta] (STT) to [bend right=15] (STGA);

    \draw[very thick, latex-latex] (TT) -- (DPT);
    \draw[very thick, latex-latex] (OA) -- (LI);
    \draw[very thick, latex-latex] (STA) -- (F);
    \draw[very thick, latex-latex] (STT) -- (TFP);

\end{tikzpicture}
}
\caption{Relation amongst the three main graphs: the triangular case.} \label{figure:graphs_triangular}
\end{figure}

\begin{theorem}[Graphs defined from networks: Triangular graphs] \label{theorem:graphs_triangular}
The diagram in Figure \ref{figure:graphs_triangular} is correct and complete.
\end{theorem}

\begin{proof}
All the one-way black implications in Figure \ref{figure:graphs_triangular}, such as ``Triangular $\TrappingGraph$ $\implies$ Triangular $\GeneralAsynchronous$'' all follow from some basic facts about the graphs; as such, we omit their proofs.

We begin by proving all the equivalences in Figure \ref{figure:graphs_triangular}.
\begin{enumerate}
    \item Locally Idempotent $\iff$ Oriented $\Asynchronous$. For all $i \in [n]$ one can decompose $f_i$ as $2^{n-1}$ functions $f_{i,a} : \B \to \B$ for all $a \in \B^{n-1}$ as follows. For any $\alpha \in \B$, let $y \in \B^n$ such that $y_i = \alpha$ and $y_{-i} = a$, then $f_{i,a}( \alpha ) = f_i( y )$. Then $f^{(i)}$ is idempotent if and only if $f_{i,a}$ is idempotent for all $a \in \B^{n-1}$, which in turn is equivalent to all arcs $( x, f^{(i)}(x) )$ in $\Asynchronous(f)$ being oriented.

    \item Trapspace-FP $\iff$ Sink-terminal $\TrappingGraph$. From Lemma \ref{lemma:sink_terminal_T}.

    \item Fixable $\iff$ Sink-terminal $\Asynchronous$. Trivial.

    \item DPT $\iff$ Triangular $\TrappingGraph$. By definition.

    \item Triangular $\TrappingGraph$ $\iff$ Oriented $\TrappingGraph$. A transitive graph is oriented if and only if it is triangular.
\end{enumerate}

We now prove all the magenta implications in Figure \ref{figure:graphs_triangular}, which hold for all trapping networks.
\begin{enumerate}[resume]
    \item Trapping and Oriented $\GeneralAsynchronous$ $\implies$ Triangular $\TrappingGraph$. In this case, $\GeneralAsynchronous = \TrappingGraph$ is transitive and oriented, and hence triangular.

    \item Trapping and Oriented $\Asynchronous$ $\implies$ Triangular $\Asynchronous$. For the sake of contradiction, suppose there is a cycle in $\Asynchronous( f )$, say $x \to \dots \to y \to x$. Then $y \in T_f( x )$ and $\HammingDistance( x, y ) = 1$, hence $x \to y$ in $\Asynchronous( f )$. Thus, the asynchronous graph is not oriented, which is the desired contradiction.

    \item Trapping and Trapspace-FP $\implies$ Fixable. By Lemma \ref{lemma:sink_terminal_T}, wWe only need to prove that if for any $x \in \B^n \setminus \Fix( f )$, there exists $y \in T_f(x)$ with $T_f( y ) \subset T_f( x )$, then $f$ is fixable. Suppose that, for the sake of contradiction, $x$ does not reach a configuration with a smaller principal trapspace. Let $A(x) \subseteq T_f(x)$ be the set of configurations reachable from $x$ in $\Asynchronous( f )$; by our hypothesis, $T_f(a) = T_f( x )$ for all $a \in A(x)$. We prove by induction on $d$ that $A( x )$ contains all the configurations $y \in T_f( x )$ at Hamming distance at most $d$ from $x$. The claim is clear for $d = 0$, hence suppose it holds for $d$. If $d \ge \HammingDistance( x, f(x) )$, there is nothing to prove, hence suppose $d \le \HammingDistance( x, f(x) ) - 1$ and let $y \in T_f( x )$ such that $\HammingDistance( x, y ) = d + 1$. Let $i \in \Delta(x, y)$, so that the configuration $a = (\neg y_i, y_{-i})$ satisfies $a \in T_f( x )$ and $\HammingDistance( a,x ) = d$, and by induction hypothesis, $a \in A(x)$. Since $N^{out}( \GeneralAsynchronous(f); a ) = T_f( x )$, we have $i \in \Delta( a, f(a) )$ and hence $a \to y$ in $\Asynchronous( f )$; thus, $y \in A(x)$ and the claim is proved. Thus, for any $y \in T_f(x)$, we have $T_f(y) = T_f(x)$, which is the desired contradiction.
\end{enumerate}

We now prove all the blue implications in Figure \ref{figure:graphs_triangular}, that hold for all commutative networks.
\begin{enumerate}[resume]
    \item Commutative and Locally Idempotent $\implies$ DPT. Anticipating a result below, let us prove that $f(y) = f(x)$ for all $y \in [x, f(x)]$ (and hence they have distinct principal trapspaces). We have $y = f^{(S)}( x )$ for some $S \subseteq [n]$, and hence
    \[
        f(y) = f^{(S, S, [n] \setminus S)}( x ) = f(x).
    \]

    \item Commutative and Fixable $\implies$ Oriented $\Asynchronous$. Thanks to Theorem \ref{theo:boolean_cs}, if $f$ is commutative and $\Asynchronous( f )$ is not oriented, then there is a connected component of $\Asynchronous( f )$ that does not contain a fixed point, and hence $f$ is not fixable.
\end{enumerate}

We finally exhibit counterexamples to implications that are not displayed in Figure \ref{figure:graphs_triangular}.
\begin{enumerate}[(a)]
    \item \label{item:counterexample_e}
    Triangular $\GeneralAsynchronous$ $\centernot\implies$ Triangular $\TrappingGraph$.

    \item \label{item:counterexample_f}
    Trapping and Triangular $\Asynchronous$ $\centernot\implies$ Oriented $\GeneralAsynchronous$.

    \item \label{item:counterexample_g}
    Oriented $\GeneralAsynchronous$ $\centernot\implies$ Sink-terminal $\TrappingGraph$.

    \item \label{item:counterexample_h}
    Sink-terminal $\TrappingGraph$ $\centernot\implies$ Sink-terminal $\GeneralAsynchronous$.

    \item \label{item:counterexample_i}
    Sink-terminal $\GeneralAsynchronous$ $\centernot\implies$ Sink-terminal $\Asynchronous$.

    \item \label{item:counterexample_j}
    Trapping and Sink-terminal $\Asynchronous$ $\centernot\implies$ Oriented $\Asynchronous$.
\end{enumerate}

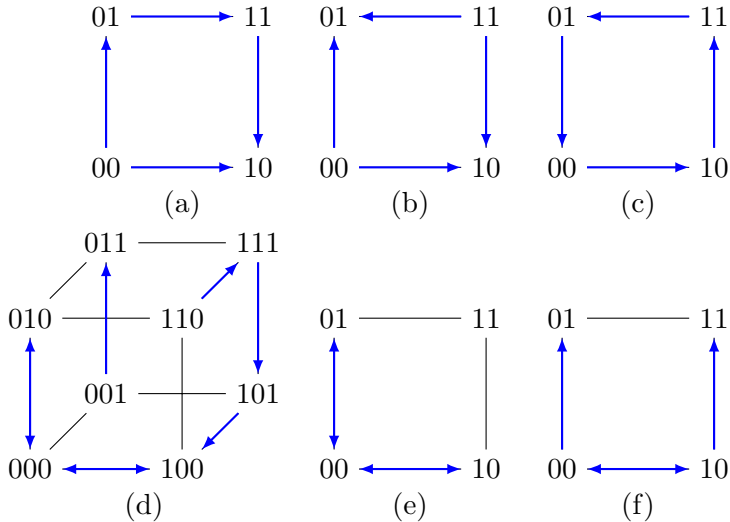
\begin{center}

\begin{tikzpicture}

\begin{scope}[xshift=0cm, yshift=0cm]
    \node (e) at (1, -0.5) {\ref{item:counterexample_e}};
    \node (00) at (0,0) {$00$};
    \node (01) at (0,2) {$01$};
    \node (10) at (2,0) {$10$};
    \node (11) at (2,2) {$11$};

    \path[draw] (00) -- (01) -- (11)
    (00) -- (10) -- (11);
    
    
    \draw[very thick,-latex, blue] (00) -- (10);
    \draw[very thick,-latex, blue] (00) -- (01);  
    \draw[very thick,-latex, blue] (01) -- (11);
    \draw[very thick,-latex, blue] (11) -- (10);
\end{scope}

\begin{scope}[xshift=5cm, yshift=0cm]
    \node (f) at (1, -0.5) {\ref{item:counterexample_f}};

    \node (00) at (0,0) {$00$};
    \node (01) at (0,2) {$01$};
    \node (10) at (2,0) {$10$};
    \node (11) at (2,2) {$11$};

    \path[draw] (00) -- (01) -- (11)
    (00) -- (10) -- (11);
    
    
    \draw[very thick,-latex, blue] (00) -- (10);
    \draw[very thick,-latex, blue] (00) -- (01);  
    \draw[very thick,-latex, blue] (11) -- (01);
    \draw[very thick,-latex, blue] (11) -- (10);
\end{scope}   

\begin{scope}[xshift=10cm, yshift=0cm]
    \node (g) at (1, -0.5) {\ref{item:counterexample_g}};

    \node (00) at (0,0) {$00$};
    \node (01) at (0,2) {$01$};
    \node (10) at (2,0) {$10$};
    \node (11) at (2,2) {$11$};

    \path[draw] (00) -- (01) -- (11)
    (00) -- (10) -- (11);
    
    
    \draw[very thick,-latex, blue] (00) -- (10);
    \draw[very thick,-latex, blue] (01) -- (00);  
    \draw[very thick,-latex, blue] (11) -- (01);
    \draw[very thick,-latex, blue] (10) -- (11);
\end{scope}   

\begin{scope}[xshift=0cm, yshift=-5cm]
    \node (h) at (1.5, -0.5) {\ref{item:counterexample_h}};

    \node (000) at (0,0) {$000$};
    \node (001) at (1,1) {$001$};
    \node (010) at (0,2) {$010$};
    \node (011) at (1,3) {$011$};
    \node (100) at (2,0) {$100$};
    \node (101) at (3,1) {$101$};
    \node (110) at (2,2) {$110$};
    \node (111) at (3,3) {$111$};

    \path[draw] (000) -- (001) -- (011) -- (111)
    (000) -- (010) -- (110) -- (111)
    (000) -- (100) -- (101) -- (111)
    (001) -- (101)
    (010) -- (011)
    (100) -- (110);
    
    \draw[very thick,latex-latex, blue] (000) -- (100);

    \draw[very thick,latex-latex, blue] (000) -- (010);
    \draw[very thick,-latex, blue] (001) -- (011);
    \draw[very thick,-latex, blue] (111) -- (101);
    
    \draw[very thick,-latex, blue] (110) -- (111);
    \draw[very thick,-latex, blue] (101) -- (100);

\end{scope}    

\begin{scope}[xshift=5cm, yshift=-5cm]
    \node (i) at (1, -0.5) {\ref{item:counterexample_i}};

    \node (00) at (0,0) {$00$};
    \node (01) at (0,2) {$01$};
    \node (10) at (2,0) {$10$};
    \node (11) at (2,2) {$11$};

    \path[draw] (00) -- (01) -- (11)
    (00) -- (10) -- (11);
    
    
    \draw[very thick,latex-latex, blue] (00) -- (10);
    \draw[very thick,latex-latex, blue] (00) -- (01);  

\end{scope} 

\begin{scope}[xshift=10cm, yshift=-5cm]
    \node (j) at (1, -0.5) {\ref{item:counterexample_j}};

    \node (00) at (0,0) {$00$};
    \node (01) at (0,2) {$01$};
    \node (10) at (2,0) {$10$};
    \node (11) at (2,2) {$11$};

    \path[draw] (00) -- (01) -- (11)
    (00) -- (10) -- (11);
    
    
    \draw[very thick,latex-latex, blue] (00) -- (10);
    \draw[very thick,-latex, blue] (00) -- (01);  
    \draw[very thick,-latex, blue] (10) -- (11);
\end{scope}   

\end{tikzpicture}
\end{center}
\end{proof}

\subsection{Globally idempotent networks} \label{subsection:globally_idempotent}

In this subsection, we give alternate definitions for globally idempotent networks.

\begin{theorem}[Alternate definitions of globally idempotent networks] \label{theorem:globally_idempotent_definitions}
Globally idempotent networks are trapping. More precisely, the following are equivalent for $f \in \Functions(n)$:
\begin{enumerate}
    \item \label{item:globally_idempotent_definition}
    $f$ is globally idempotent, i.e. $f^{(S,S)} = f^{(S)}$ for all $S \subseteq [n]$.

    \item \label{item:globally_idempotent_xy}
    $[y, f(y)] \subseteq [y, f(x)]$ for all $x \in \B^n$ and any $y \in [x, f(x)]$.

    \item \label{item:globally_idempotent_ST}
    $f^{(S, T)} \sqsupseteq f^{(S \cap T)}$ for all $S, T \subseteq [n]$.
    
\end{enumerate}
\end{theorem}

\begin{proof}
$\ref{item:globally_idempotent_definition} \implies \ref{item:globally_idempotent_xy}$. Let $x \in \B^n$ and $y \in [x, f(x)]$. For any $i \in \Delta( y, f(y) )$, we have $y_i \ne f_i(y) = f_i(x)$, hence $i \in \Delta( y, f(x) )$. 

$\ref{item:globally_idempotent_xy} \implies \ref{item:globally_idempotent_definition}$. Suppose, for the sake of contradiction, that there exist $S \subseteq [n]$, $i \in S$, and $x \in \B^n$ such that $f_i( f^{(S)}(x) ) \ne f_i( x )$. Since $f^{(S)}( x ) = f^{(S \cap \Delta(x, f(x)))}( x )$, we can assume $S \subseteq \Delta( x, f(x) )$. Denoting $y = f^{(S)}( x )$, we have $y \in [x, f(x)]$, $y_i = f_i( x ) = \neg f_i( y )$. Thus, $i \in \Delta( y, f(y) ) \subseteq \Delta( y, f(x) )$, which is the desired contradiction.

$\ref{item:globally_idempotent_definition} \implies \ref{item:globally_idempotent_ST}$. 
$f$ is trapping, hence $f^{( S,T )} \sqsubseteq f^{( S \cup T )}$. Moreover,
\[
    f^{(S, T)} = f^{( S, T \setminus S)},
\]
and hence $f^{(S, T)}_{ S \cap T } = f_{S \cap T}$, which implies $f^{ (S \cap T) } \sqsubseteq f^{( S,T )}$.

$\ref{item:globally_idempotent_ST} \implies \ref{item:globally_idempotent_definition}$. We have $f^{(S)} = f^{(S, S)}$ for all $S \subseteq [n]$, i.e. $f$ is globally idempotent.
\end{proof}

\begin{theorem}[Classification of Lille networks] \label{theorem:Lille_classification}
The diagram in Figure \ref{figure:Lille} is correct and complete.
\end{theorem}

\begin{proof}
Once again, all the unidirectional black implications are either trivial or follow previous results; as such, their proofs are omitted. 

We now prove the magenta implications, that only hold for trapping networks (and that are not direct consequences of their counterparts in Figure \ref{figure:graphs_triangular}).
\begin{enumerate}
    \item Trapping and Trapspace-FP $\implies$ Interval-FP. Every interval is a principal trapspace, and hence it contains a fixed point.

    \item Trapping and Interval-UFP Idempotent $\implies$ Lille. For all $x \in \B^n$, $f(x)$ is the unique fixed point in $[x, f(x)] = T_f(x)$. Let $y \in [x, f(x)]$, then $f(y)$ is the unique fixed point in $T_f(y) \subseteq T_f(x)$, hence $f(y) = f(x)$. Thus, $[y, f(y)] = [x, f(x)]$ for all $x \in \B^n$ and $y \in [x, f(x)]$.
\end{enumerate}

All the blue implications, that only hold for commutative networks, follow from the implication below.
\begin{enumerate}[resume]
    \item Commutative and Trapspace-FP $\implies$ Lille. By Theorem \ref{theorem:graphs_triangular}, $f$ is commutative and locally idempotent, i.e. $f$ is Lille.
\end{enumerate}

\begin{figure}
\centering
\resizebox{\textwidth}{!}{
\begin{tikzpicture}[yscale=2, xscale=4, font=\small]

    
    \node[draw, very thick, red, minimum height=1cm, minimum width=2.8cm, align=center] (IC) at (0,1) {Lille}; 
    \node[draw, minimum height=1cm, minimum width=2.8cm, align=center] (UFP) at (0.5,0) {Interval-UFP \\ Idempotent};

    \node[draw, minimum height=1cm, minimum width=2.8cm, align=center, orange] (GI) at (1,1) {Globally \\ Idempotent};
    \node[draw,  minimum height=1cm, minimum width=2.8cm, align=center] (UFPI) at (1,-1) {Interval-UFP};
    \node[draw,  minimum height=1cm, minimum width=2.8cm, align=center] (I) at (1.5,0) {Idempotent};

    \node[draw, minimum height=1cm, minimum width=2.8cm, align=center] (DPT) at (1.5,2) {DPT};

    \node[draw, minimum height=1cm, minimum width=2.8cm, align=center] (TGA) at (2,1) {Triangular \\ $\GeneralAsynchronous$};

    \node[draw, minimum height=1cm, minimum width=2.8cm, align=center] (TA) at (2.5,0) {Triangular \\ $\Asynchronous$};
    \node[draw, minimum height=1cm, minimum width=2.8cm, align=center] (OGA) at (2.5,2) {Oriented \\ $\GeneralAsynchronous$};

    \node[draw, minimum height=1cm, minimum width=2.8cm, align=center] (OA) at (3,1) {Locally \\ Idempotent};
    
    \node[draw, minimum height=1cm, minimum width=2.8cm, align=center] (F) at (3.5,0) {Fixable};

    \node[draw, minimum height=1cm, minimum width=2.8cm, align=center] (FI) at (2,-1) {Interval-FP};

    \node[draw, minimum height=1cm, minimum width=2.8cm, align=center] (FPT) at (4,-1) {Trapspace-FP};


    \draw[thick, -latex] (IC) -- (UFP);
    \draw[very thick, -latex, magenta] (UFP) to [bend right=15] (IC);

    \draw[thick, -latex] (UFP) -- (UFPI);
    \draw[thick, -latex, blue] (UFPI) to [bend right=15] (UFP);
    
    \draw[thick, -latex] (UFPI) -- (FI);
    \draw[thick, -latex, blue] (FI) to [bend right=15] (UFPI);

    \draw[thick, -latex] (IC) -- (GI);
    \draw[thick, -latex, blue] (GI) to [bend right=15] (IC);

    \draw[thick, -latex] (DPT) -- (TGA);
    \draw[very thick, -latex, magenta] (TGA) to [bend right=15] (DPT);

    \draw[thick, -latex] (GI) -- (I);
    \draw[thick, -latex, blue] (I) to [bend right=15] (GI);
    
    \draw[thick, -latex] (GI) -- (DPT);
    \draw[thick, -latex, blue] (DPT) to [bend right=15] (GI);

    \draw[thick, -latex] (TGA) -- (TA);
    \draw[thick, -latex, blue] (TA) to [bend right=15] (TGA);

    \draw[thick, -latex] (TGA) -- (OGA);
    \draw[very thick, -latex, magenta] (OGA) to [bend right=15] (TGA);

    \draw[thick, -latex] (TA) -- (OA);
    \draw[very thick, -latex, magenta] (OA) to [bend right=15] (TA);

    \draw[thick, -latex] (OGA) -- (OA);
    \draw[thick, -latex, blue] (OA) to [bend right=15] (OGA);


    \draw[thick, -latex] (TA) -- (F);
    \draw[thick, -latex, blue] (F) to [bend right=15] (TA);

    \draw[thick, -latex] (F) -- (FPT);
    \draw[very thick, -latex, magenta] (FPT) to [bend right=15] (F);


    \draw[thick, -latex] (UFP) -- (I);
    \draw[thick, -latex, blue] (I) to [bend right=15] (UFP);

    \draw[thick, -latex] (I) -- (FI);
    \draw[thick, -latex, blue] (FI) to [bend right=15] (I);


    \draw[thick, -latex] (FI) -- (FPT);
    \draw[very thick, -latex, magenta] (FPT) to [bend right=15] (FI);

\end{tikzpicture}
}
\caption{\textcolor{red}{Lille} v \textcolor{blue}{commutative} and \textcolor{magenta}{trapping} networks.} \label{figure:Lille}
\end{figure}
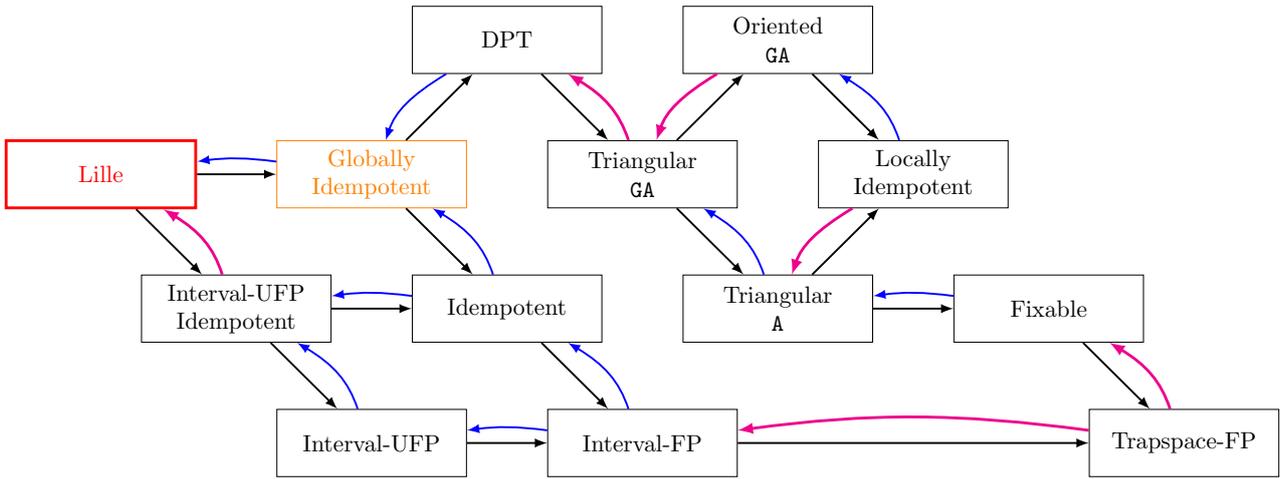

We now exhibit counterexamples to implications that are not displayed in Figure \ref{figure:Lille}.
\begin{enumerate}[(a)]
    \item \label{item:counterexample_k} Interval-UFP and Idempotent $\centernot\implies$ Fixable.

    \item \label{item:counterexample_l} Interval-UFP and Idempotent $\centernot\implies$ Locally Idempotent.

    \item \label{item:counterexample_m} Trapping and Interval-UFP $\centernot\implies$ Idempotent.

    \item \label{item:counterexample_n} Globally Idempotent $\centernot\implies$ Interval-UFP.
    
    \item \label{item:counterexample_o} Trapping and DPT $\centernot\implies$ Idempotent.
    
    \item \label{item:counterexample_p} Trapping and Idempotent $\centernot\implies$ Locally Idempotent.
    
    \item \label{item:counterexample_q} Trapping and Interval-UFP $\centernot\implies$ Locally Idempotent.
\end{enumerate}

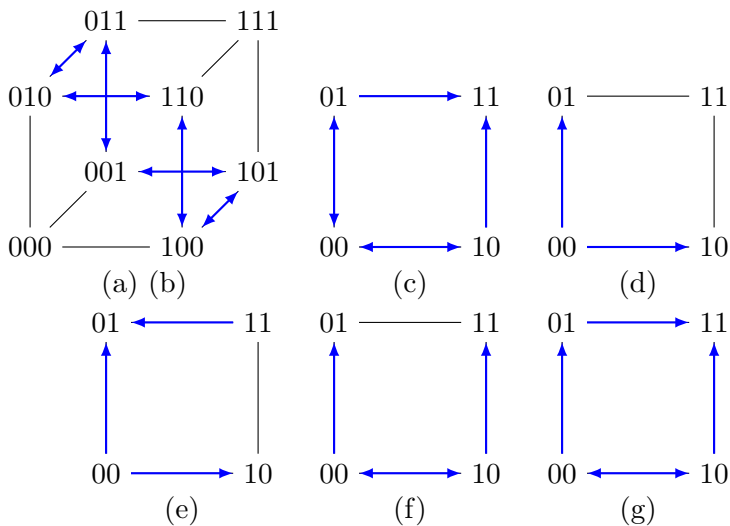
\begin{center}
    
\begin{tikzpicture}

\begin{scope}[xshift=0cm, yshift=0cm]
    \node (kl) at (1.5, -0.5) {\ref{item:counterexample_k} \ref{item:counterexample_l}};
    \node (000) at (0,0) {$000$};
    \node (001) at (1,1) {$001$};
    \node (010) at (0,2) {$010$};
    \node (011) at (1,3) {$011$};
    \node (100) at (2,0) {$100$};
    \node (101) at (3,1) {$101$};
    \node (110) at (2,2) {$110$};
    \node (111) at (3,3) {$111$};

    \path[draw] (000) -- (001) -- (011) -- (111)
    (000) -- (010) -- (110) -- (111)
    (000) -- (100) -- (101) -- (111)
    (001) -- (101)
    (010) -- (011)
    (100) -- (110);
    
    \draw[very thick,latex-latex, blue] (001) -- (101);
    \draw[very thick,latex-latex, blue] (001) -- (011);

    \draw[very thick,latex-latex, blue] (010) -- (110);
    \draw[very thick,latex-latex, blue] (010) -- (011);

    \draw[very thick,latex-latex, blue] (100) -- (110);
    \draw[very thick,latex-latex, blue] (100) -- (101);
\end{scope}

\begin{scope}[xshift=5cm, yshift=0cm]
    
    \node (m) at (1, -0.5) {\ref{item:counterexample_m}};
    \node (00) at (0,0) {$00$};
    \node (01) at (0,2) {$01$};
    \node (10) at (2,0) {$10$};
    \node (11) at (2,2) {$11$};

    \path[draw] (00) -- (01) -- (11)
    (00) -- (10) -- (11);
    
    
    \draw[very thick,latex-latex, blue] (00) -- (10);
    \draw[very thick,latex-latex, blue] (00) -- (01);  
    \draw[very thick,-latex, blue] (10) -- (11);
    \draw[very thick,-latex, blue] (01) -- (11);
\end{scope}
 
\begin{scope}[xshift=10cm, yshift=0cm]
    \node (n) at (1, -0.5) {\ref{item:counterexample_n}};
    \node (00) at (0,0) {$00$};
    \node (01) at (0,2) {$01$};
    \node (10) at (2,0) {$10$};
    \node (11) at (2,2) {$11$};

    \path[draw] (00) -- (01) -- (11)
    (00) -- (10) -- (11);
    
    
    \draw[very thick,-latex, blue] (00) -- (10);
    \draw[very thick,-latex, blue] (00) -- (01);  
\end{scope}  

\begin{scope}[xshift=0cm, yshift=-4cm]
    \node (o) at (1, -0.5) {\ref{item:counterexample_o}};
    \node (00) at (0,0) {$00$};
    \node (01) at (0,2) {$01$};
    \node (10) at (2,0) {$10$};
    \node (11) at (2,2) {$11$};

    \path[draw] (00) -- (01) -- (11)
    (00) -- (10) -- (11);
    
    
    \draw[very thick,-latex, blue] (00) -- (10);
    \draw[very thick,-latex, blue] (00) -- (01);  
    \draw[very thick,-latex, blue] (11) -- (01);  
\end{scope} 

\begin{scope}[xshift=5cm, yshift=-4cm]
    \node (p) at (1, -0.5) {\ref{item:counterexample_p}};
    \node (00) at (0,0) {$00$};
    \node (01) at (0,2) {$01$};
    \node (10) at (2,0) {$10$};
    \node (11) at (2,2) {$11$};

    \path[draw] (00) -- (01) -- (11)
    (00) -- (10) -- (11);
    
    
    \draw[very thick,latex-latex, blue] (00) -- (10);
    \draw[very thick,-latex, blue] (00) -- (01);  
    \draw[very thick,-latex, blue] (10) -- (11);  
\end{scope}

\begin{scope}[xshift=10cm, yshift=-4cm]
    \node (q) at (1, -0.5) {\ref{item:counterexample_q}};
    \node (00) at (0,0) {$00$};
    \node (01) at (0,2) {$01$};
    \node (10) at (2,0) {$10$};
    \node (11) at (2,2) {$11$};

    \path[draw] (00) -- (01) -- (11)
    (00) -- (10) -- (11);
    
    
    \draw[very thick,latex-latex, blue] (00) -- (10);
    \draw[very thick,-latex, blue] (00) -- (01);  
    \draw[very thick,-latex, blue] (10) -- (11);  
    \draw[very thick,-latex, blue] (01) -- (11);  
\end{scope}

\end{tikzpicture}
\end{center}

\end{proof}

\section{Conclusion}

\paragraph{Summary of results}

This paper tied together the topics of trapspaces of Boolean networks and of commutative networks.
We have introduced trapping networks, that generalise commutative networks and are a normal form for the study of trapspaces of networks.
We have also classified the collections of trapspaces and of principal trapspaces of networks.
We have then focused on two particular classes of commutative networks, namely Marseille (bijective commutative) and Lille (idempotent commutative) networks.
Those two classes are very well structured (for instance, one may extract over twenty different definitions of Lille networks in the paper) and highlight the broader structure of commutative and trapping networks at large.

\paragraph{Future work}

We identify three main avenues for potential future work. First, one may want to develop the theory of trapping networks. In particular, one may classify networks of a particular class (e.g. linear, monotone, increasing, or with particular interaction graphs) that are trapping. Second, we have classified the collections of principal trapspaces as the so-called pre-principal collections of subcubes, which have a ``simple'' definition (i.e. $\mathcal{Q} = \mu( \mathcal{Q} )$). However, the computational complexity of deciding whether a collection of subcubes is pre-principal remains unknown. The same holds for collections of trapspaces, which are the pre-ideal collections of subcubes. Third, trapping networks are one way of generalising commutative networks; Table \ref{table:characterisation} naturally suggests other ways of doing so. It would be interesting to see how different generalisations of commutativity behave, both in terms of their dynamical properties and whether classifications such as in Figures \ref{figure:Marseille} and \ref{figure:Lille} can be derived.

\section*{Acknowledgment}

I would like to thank Lo\"ic Paulev\'e and Sara Riva for developing work on trapspaces that led to the topic of this paper, and for fruitful discussions and advice throughout the preparation of this paper.


\end{document}